\documentclass[11pt]{article}
\usepackage{amsmath,amssymb,amsthm}
\usepackage{enumitem}
\usepackage{fullpage}
\usepackage{hyperref}
\usepackage{mathabx}
\usepackage{blkarray}
\usepackage{tikz}
\usepackage{authblk}
\usepackage{algorithmicx,algorithm}
\usepackage{algpseudocode}
\usepackage{lineno}
\usepackage{todonotes}
\usepackage{thm-restate}
\usetikzlibrary{backgrounds}
\usetikzlibrary{shapes}
\usetikzlibrary{positioning}
\hypersetup{
	colorlinks,
	linkcolor={red!60!black},
	citecolor={green!60!black},
	urlcolor={blue!60!black},
	pdftitle = {A polynomial kernel for 3-leaf power deletion},
	pdfauthor = {Jungho Ahn, Eduard Eiben, O-joung Kwon, and Sang-il Oum}
}
\newcommand\abs[1]{\lvert #1\rvert}
\newtheorem{THM}{Theorem}[section]
\newtheorem{LEM}[THM]{Lemma}

\newtheorem{PROP}[THM]{Proposition}

\newtheorem*{PROPredundant}{Proposition~\ref{prop:redundant}}
\newtheorem{CLM}{Claim}

\theoremstyle{remark}

\theoremstyle{definition}
\newtheorem*{RULES}{Reduction Rules}
\newtheorem{RULE}{Reduction Rule}

\newenvironment{subproof}[1][\proofname]{
	
	\begin{proof}[#1]}{\end{proof}}

\newcommand{\NP}{\mbox{{\textsf{NP}}}}
\newcommand{\VC}{\textsc{Vertex Cover}}
\newcommand{\CL}{\textsc{Cluster Vertex Deletion}}
\newcommand{\LPD}{\textsc{$3$-Leaf Power Deletion}}

\usepackage{boxedminipage}

\begin{document}
\title{A polynomial kernel for $3$-leaf power deletion}
\footnotetext{The paper is published in Algorithmica~\cite{AEKO-journal} and an extended abstract appeared in the Proceedings of the 45th International Symposium on Mathematical Foundations of Computer Science~\cite{AEKO-MFCS2020}.}
\author[$\ast$2,1]{Jungho~Ahn}
\author[3]{Eduard~Eiben}
\author[4,1]{O-joung~Kwon\thanks{Supported by the Institute for Basic Science (IBS-R029-C1).}\thanks{Supported by the National Research Foundation of Korea (NRF) grant funded by the Ministry of Education (No. NRF-2018R1D1A1B07050294 and No. NRF-2021K2A9A2A11101617).}}
\author[$\ast$1,2]{Sang-il~Oum}
\affil[1]{Discrete Mathematics Group, Institute for Basic Science (IBS), Daejeon, South~Korea}
\affil[2]{Department of Mathematical Sciences, KAIST, Daejeon,~South~Korea}
\affil[3]{Department of Computer Science, Royal Holloway, University~of~London, Egham, United Kingdom}
\affil[4]{Department of Mathematics, Hanyang University, Seoul,~South~Korea}
\affil[ ]{\small \textit{Email addresses:}  \texttt{junghoahn@kaist.ac.kr},
  \texttt{Eduard.Eiben@rhul.ac.uk}, \texttt{sangil@ibs.re.kr},
  \texttt{ojoungkwon@hanyang.ac.kr}}

\maketitle
\begin{abstract}
	For a non-negative integer $\ell$, the \emph{$\ell$-leaf power of a tree~$T$} is a simple graph $G$ on the leaves of $T$ such that two vertices are adjacent in $G$ if and only if their distance in~$T$ is at most $\ell$.
	We provide a polynomial kernel for the problem of deciding whether we can delete at most $k$ vertices to make an input graph a $3$-leaf power of some tree.
	More specifically, we present a polynomial-time algorithm for an input instance $(G,k)$ for the problem to output an equivalent instance $(G',k')$ such that $k'\leq k$ and $G'$ has at most $O(k^{14})$ vertices.
\end{abstract}

\section{Introduction}\label{sec:intro}

For a non-negative integer $\ell$, the \emph{$\ell$-leaf power of a tree~$T$} is a simple graph~$G$ on the leaves of $T$ such that two vertices $v$ and $w$ are adjacent in $G$ if and only if the distance between $v$ and~$w$ in~$T$ is at most~$\ell$.
For brevity, we say that a graph is an $\ell$-leaf power if it is the $\ell$-leaf power of some tree.
Nishimura, Ragde, and Thilikos \cite{NRT2002} introduced $\ell$-leaf powers to understand the structure of phylogenetic trees in computational biology, as one can model the structure of `similarity graphs' among species as the $\ell$-leaf power of a phylogenetic tree.

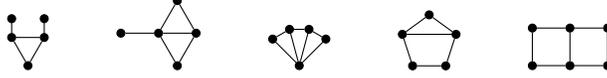
\begin{figure}
	\centering
	\tikzstyle{v}=[circle, draw, solid, fill=black, inner sep=0pt, minimum width=3pt]
	\begin{tikzpicture}[scale=0.75]
	\draw (0,-0.55) node(){};
	\foreach \x in {0,1,2} {
		\draw (\x*120-90:.5) node [v](v\x){};
	}
	\draw (v1)--(v2)--(v0)--(v1);
	\draw (v1) --+(0,0.75)node[v]{};
	\draw (v2) --+(0,0.75)node[v]{};
	\end{tikzpicture}
	\qquad\quad
	\begin{tikzpicture}[scale=0.75]
	\foreach \x in {-1,0,1,3}{
		\draw (60*\x:1) node[v](v\x){};
	}
	\node[v](c){};
	\draw(v-1)--(v0)--(v1)--(c)--(v0);
	\draw(v3)--(c)--(v-1);
	\end{tikzpicture}
	\qquad\quad
	\begin{tikzpicture}[scale=0.75]
	\draw (0,-0.27) node(){};
	\draw (0,-0.15) node[v](c){};
	\foreach \x in {0,1,2,3} {
		\draw (c)-- (30+40*\x:1) node[v] (x\x){};
	}
	\draw (x0)--(x1)--(x2)--(x3);
	\end{tikzpicture}
	\qquad\quad
	\begin{tikzpicture}[scale=0.75]
	\foreach \x in {0,1,2,3,4} {
		\draw (72*\x+90:0.8) node[v](x\x){};
	}
	\draw (x0)--(x1)--(x2)--(x3)--(x4)--(x0);
	\draw (x1)--(x4);
	\end{tikzpicture}
	\qquad\quad
	\begin{tikzpicture}[scale=0.75]
	\draw (0,-0.20) node(){};
	\foreach \y in {0,1}{
		\foreach \x in {0,1,2}{
			\draw (\x,\y) node[v](v\x\y){};
		}
		\draw (v0\y)--(v1\y)--(v2\y);
	}
	\foreach \x in {0,1,2}{
		\draw (v\x0)--(v\x1);
	}
	\end{tikzpicture}
	
	\caption{A bull, a dart, a gem, a house, and a domino.}
	\label{fig:1}
\end{figure}

There have been studies on structures and recognition algorithms for $\ell$-leaf powers.
It is easy to see that $0$-leaf powers are edgeless graphs, $1$-leaf powers are $K_2$ and edgeless graphs, and $2$-leaf powers are the disjoint union of complete graphs.
Dom, Guo, H\"{u}ffner, and Niedermeier~\cite{DGHL2006} showed that $3$-leaf powers are precisely (bull, dart, gem)-free chordal graphs, where chordal graphs are graphs with no chordless cycles of length at least $4$; see Figure~\ref{fig:1} for an illustration of a bull, a dart, and a gem.
There are linear-time algorithms to recognize $\ell$-leaf powers~\cite{BLS2008,CK2007} for $\ell\in\{4,5\}$ and a polynomial-time algorithm to recognize $6$-leaf powers~\cite{Ducoffe2019}.
Eppstein and Havvaei~\cite{EH2019} presented a linear-time algorithm to recognize $k$-degenerate $\ell$-leaf powers for fixed~$k$ and~$\ell$.
Very recently, Lafond~\cite{Lafond2022} presented an algorithm that recognize $\ell$-leaf powers in polynomial time for each fixed~$\ell$.

We are interested in the following \emph{vertex deletion} problem, which generalizes the corresponding recognition problem.

\medskip
\noindent
\fbox{\parbox{0.97\textwidth}{
	\textsc{$\ell$-Leaf Power Deletion}
	\begin{description}
		\item[Input:] A graph $G$ and a non-negative integer $k$
		\item[Parameter:] $k$
		\item[Question:] Is there a set $S\subseteq V(G)$ with $\abs{S}\leq k$ such that $G\setminus S$ is an $\ell$-leaf power?
	\end{description}}}
\medskip

Vertex deletion problems include some of the best studied \NP-hard problems in theoretical computer science, including \textsc{Vertex Cover} and \textsc{Feedback Vertex Set}.
In general, the problem asks whether it is possible to delete at most $k$ vertices from an input graph so that the resulting graph belongs to a specified graph class.
Lewis and Yannakakis \cite{LY1980} showed that every vertex deletion problem to a non-trivial\footnote{A class of graphs is \emph{non-trivial} if both itself and its complement contain infinitely many non-isomorphic graphs.} and hereditary\footnote{A class of graphs is \emph{hereditary} if it is closed under taking induced subgraphs.} graph class is \NP-hard.
Since the class of $\ell$-leaf powers is non-trivial and hereditary for every non-negative integer $\ell$, it follows that the \textsc{$\ell$-Leaf Power Deletion} problem is \NP-hard.

Vertex deletion problems have been investigated on various graph classes through the parameterized complexity paradigm~\cite{CyganFKLMPPS15,DF2013}, which measures the performance of algorithms not only with respect to the input size but also with respect to additional numerical parameters.
The notion of vertex deletion allows a highly natural choice of the parameter, specifically the size~$k$ of the deletion set.
A parameterized problem $\Pi$ is \emph{fixed-parameter tractable} if it can be solved by an algorithm with running time $f(k)\cdot n^{O(1)}$ for the input size $n$ and some computable function $f:\mathbb{N}\to\mathbb{N}$.
It is well known that~$\Pi$ is fixed-parameter tractable if and only if it admits a kernel~\cite{DF2013}.
A \emph{kernel} is a polynomial-time preprocessing algorithm that transforms the given instance of the problem into an equivalent instance whose size is bounded from above by some function $f(k)$ of the parameter.
The function $f(k)$ is usually referred to as the \emph{size} of the kernel.
A \emph{polynomial kernel} is then a kernel with size bounded from above by some polynomial in~$k$.
For a fixed-parameter tractable problem, one of the most natural follow-up questions in parameterized complexity is whether the problem admits a polynomial kernel.
The existence of polynomial kernels for vertex deletion problems has been widely investigated; see Fomin, Lokshtanov, Saurabh, and Zehavi~\cite{kernelization}.

\begin{table}
	\centering
	\makebox[\linewidth]{
	\begin{tabular}{c|c|c|c}
		$\ell$ & Running time & \begin{tabular}{@{}c@{}} Kernel \\ (The number of vertices)\end{tabular} & Remark \\
		\hline
		$0$ & $O(1.2738^k+kn)$~\cite{CKX2010} & $2k-\Omega(\log k)$~\cite{Lampis2011,LPcure} & Equivalent to \textsc{VC} \\
		\hline
		$1$ & $O(1.2738^k+kn)$~\cite{CKX2010} & $2k-\Omega(\log k)$~\cite{Lampis2011,LPcure} & Reduced to \textsc{VC} \\
		\hline
		$2$ & $O(2^kk\cdot m\sqrt{n}\log n)$~\cite{HKMN2010} & $O(k^{5/3})$~\cite{cluster_kernel} & Equivalent to \textsc{CD} \\
		\hline
		$3$ & $O(37^k\cdot n^7(n+m))$~\cite{Thesis-Ahn} & $O(k^{14})$ [Theorem~\ref{thm:main}] & - \\
	\end{tabular}}
	\caption{The current best known running time of fixed-parameter algorithms and upper bounds for the number of vertices in kernels for \textsc{$\ell$-Leaf Power Deletion} with small~$\ell$. For an input graph $G$, let $n:=\abs{V(G)}$ and $m:=\abs{E(G)}$. \textsc{VC} and \textsc{CD} stand for \VC\ and \CL, respectively.}
	\label{tab:results}
\end{table}

We are going to survey known results of the \textsc{$\ell$-Leaf Power Deletion} problems for small values of $\ell$; see Table~\ref{tab:results}.
Let $(G,k)$ be an instance of the \textsc{$\ell$-Leaf Power Deletion} problem.

The \textsc{$0$-Leaf Power Deletion} problem is identical to \VC.
Currently, the best known fixed-parameter algorithm for \VC\ runs in time $O(1.2738^k+k\abs{V(G)})$ by Chen, Kanj, and Xia~\cite{CKX2010}, and $2k-\Omega(\log k)$ is the best known upper bound for the number of vertices in kernels for \VC, independently by Lampis~\cite{Lampis2011} and Lokshtanov, Narayanaswamy, Raman, Ramanujan, and Saurabh~\cite{LPcure}.

Since $1$-leaf powers are $K_2$ or edgeless graphs, one can easily reduce the \textsc{$1$-Leaf Power Deletion} problem to \VC, so that it can be solved in time $O(1.2738^k+k\abs{V(G)})$ and admits a kernel with $2k-\Omega(\log k)$ vertices.

The \textsc{$2$-Leaf Power Deletion} problem was studied under the name of the \CL\ problem.
H\"{u}ffner, Komusiewicz, Moser, and Niedermeier~\cite{HKMN2010} showed that the \CL\ problem is fixed-parameter tractable by presenting an algorithm with running time $O(2^kk\cdot\abs{E(G)}\sqrt{\abs{V(G)}}\log\abs{V(G)})$, and Fomin, Le, Lokshtanov, Saurabh, Thomass\'{e}, and Zehavi~\cite{cluster_kernel} presented a kernel with $O(k^{5/3})$ vertices for the problem.

Our aim is to investigate the situation for $\ell=3$.
Dom, Guo, H\"{u}ffner, and Niedermeier~\cite{DGHL2006} already showed that the \LPD\ problem is fixed-parameter tractable.
The algorithm in~\cite{EGK2018} can be modified to a \emph{single-exponential} fixed-parameter algorithm for the \LPD\ problem that runs in time $\alpha^k\cdot n^{O(1)}$ for the input size~$n$ and some constant $\alpha>1$; see~\cite{Thesis-Ahn}.
Here is our main theorem.

\begin{restatable}{THM}{main}\label{thm:main}
	The \LPD\ problem admits a kernel with $O(k^{14})$ vertices.
\end{restatable}

We remark that for a class $\mathcal C$ of graphs, it is possible that $\mathcal C$ admits a polynomial kernel for the vertex deletion problem to~$\mathcal{C}$ and yet a subclass of~$\mathcal C$ does not.
For instance, using the well-known sunflower lemma, it is easy to show that the vertex deletion problem to the class of graphs having no induced cycle of length $5$ admits a polynomial kernel; see~\cite{kernelization}.
However, unless FPT=W[2], the vertex deletion problem to the class of perfect graphs, which is a subclass of the previous class, does not admit a kernel, because the problem is W[2]-hard~\cite{perfect2013}.
Thus, we cannot simply derive our results from known polynomial kernels on vertex deletion problems to classes of chordal graphs~\cite{JP2018,ALMSZ2019} or distance-hereditary graphs~\cite{KK2021}, each of which is a superclass of the class of $3$-leaf powers.

Let us sketch our method.
The first step is to find a \emph{``good''} approximate solution.
For an input graph $G$, we will find a \emph{$1$-redundant modulator}, of size $O(k^2)$, that is a set $S\subseteq V(G)$ such that $G\setminus(S\setminus\{v\})$ is a $3$-leaf power for every vertex $v\in S$.
The technique of computing a $1$-redundant modulator is quite standard in several kernelization algorithms~\cite{JP2018,KK2017-block,KK2021,intervalkernel}; after finding an approximate solution via known approximation algorithms, we attach a few more vertices to the solution so that it becomes $1$-redundant.
Since $\abs{S}=O(k^2)$, to obtain a polynomial kernel, it suffices to reduce the size of~$G\setminus S$.
To do this, we reduce the number of components of $G\setminus S$, and then reduce the size of each component of~$G\setminus S$.

To reduce the number of components, we will classify components of $G\setminus S$.
For a component~$C$ of $G\setminus S$, if $N_G(V(C))\cup N_G(V(G)\setminus V(C))$ is a clique, then the partition $(V(C),V(G)\setminus V(C))$ will be called a complete split of $G$ and it will be easier to handle such components.
If it does not form a complete split, then we will find a pair of vertices in $N_G(V(C))\subseteq S$ certifying that $(V(C),V(G)\setminus V(C))$ is not a complete split.
Such a pair will be called a blocking pair and we say that $V(C)$ is blocked by the pair.
A key property, Lemma~\ref{lem:obs from blocking}, of a blocking pair is that 
for two components $C_1$ and~$C_2$ of $G\setminus S$, 
if both $V(C_1)$ and $V(C_2)$ are blocked by the same pair $\{x,y\}$, 
then $V(C_1)\cup V(C_2)\cup\{x,y\}$ induces a subgraph that is not a $3$-leaf power.
Through a marking process with pairs in $S$, we show that if there are many components of $G\setminus S$ blocked by some pairs, then we can safely remove all edges inside some of the components.
Afterward, we bound the number of isolated vertices of $G\setminus S$ through another marking process, and then design a series of reduction rules to bound the size of the remaining components of $G\setminus S$ by using a tree-like structure of $3$-leaf powers by Brandst\"{a}dt and Le~\cite{BL2006}.

We organize this paper as follows.
In Section~\ref{sec:prelim}, we summarize some terminologies in graph theory and parameterized complexity, and introduce known properties of $3$-leaf powers.
In Section~\ref{sec:overview}, we describe our polynomial kernel without proofs.
In Section~\ref{sec:redundant}, we present a polynomial-time algorithm that finds a small $1$-redundant modulator.
In Section~\ref{sec:safeness}, we show the safeness of all reduction rules for the \LPD\ problem.
In Section~\ref{sec:main proof}, we prove Theorem~\ref{thm:main}, and in Section~\ref{sec:conclu}, we conclude this paper with some open problems.

\section{Preliminaries}\label{sec:prelim}

In this paper, all graphs are finite and simple.
For every vertex $v$, we denote by~$N_G(v)$ the set of neighbors of~$v$ in~$G$.
For every set $X\subseteq V(G)$, we denote by~$N_G(X)$ the set of vertices in $V(G)\setminus X$ that are adjacent to some vertices in~$X$.
We let $N_G[X]:=N_G(X)\cup X$.
We may omit the subscripts of these notations if it is clear from the context.
For disjoint subsets $X$ and~$Y$ of $V(G)$, we say that~$X$ is \emph{complete} to $Y$ if each vertex in~$X$ is adjacent to all vertices in~$Y$, and \emph{anti-complete} to $Y$ if each vertex in~$X$ is non-adjacent to all vertices in~$Y$.
Note that if~$X$ is complete to~$Y$, then~$Y$ is complete to~$X$
and if~$X$ is anti-complete to~$Y$, then~$Y$ is anti-complete to~$X$.
We denote by $G\setminus X$ the graph obtained from~$G$ by removing all vertices in~$X$ and all edges incident with some vertices in~$X$, and let $G[X]:=G\setminus(V(G)\setminus X)$.
We may write $G\setminus v$ instead of $G\setminus\{v\}$ for a vertex $v$ of $G$.
For a set $T\subseteq E(G)$, let $G\setminus T$ be the graph obtained from $G$ by removing all edges in $T$.

A graph $G$ is \emph{trivial} if $\abs{V(G)}\leq1$, and \emph{non-trivial}, otherwise.
A \emph{clique} is a set of pairwise adjacent vertices.
A graph is \emph{complete} if its vertex set is a clique, and \emph{incomplete}, otherwise.
An \emph{independent set} is a set of pairwise non-adjacent vertices.
Distinct vertices $v$ and $w$ of $G$ are \emph{twins} in $G$ if $N_G(v)\setminus\{w\}=N_G(w)\setminus\{v\}$.
Twins $v$ and $w$ in $G$ are \emph{true} if they are adjacent, and \emph{false}, otherwise.
A \emph{twin-set} in $G$ is a set of pairwise twins in~$G$.
Note that every twin-set is a clique or an independent set.
A twin-set is \emph{true} if it is a clique, and \emph{false} otherwise.

A vertex $v$ of a graph $G$ is a \emph{cut-vertex} of $G$ if $G\setminus v$ has more components than~$G$.
A set $X\subseteq V(G)$ is a \emph{vertex cut} of $G$ if $G\setminus X$ has more components than~$G$.
A \emph{clique cut-set} is a vertex cut which is a clique.
A set $Y\subseteq V(G)$ is a \emph{feedback vertex set} of $G$ if $G\setminus Y$ is a forest.

A vertex of a graph is \emph{isolated} if it has no neighbors.
Vertices of a tree are also called \emph{nodes}.
A node of a tree is a \emph{leaf} if it has exactly one neighbor, and is \emph{branching} if it has at least three neighbors.
For graphs $G_1,\ldots,G_m$, a graph is \emph{$(G_1,\ldots,G_m)$-free} if it has no induced subgraph isomorphic to $G_i$ for all $i=1,\ldots,m$.

We say that a reduction rule is \emph{safe} if every input instance of the \LPD\ problem is equivalent to the resulting instance obtained by applying the reduction rule to the input instance.

\subsection{Parameterized problems and kernels}

For a fixed finite set $\Sigma$ of alphabets, an \emph{instance} is an element in $\Sigma^*\times\mathbb{N}$.
For an instance $(I,k)$, we call $k$ a \emph{parameter}.
A \emph{parameterized problem} is a set $\Pi\subseteq\Sigma^*\times\mathbb{N}$.
A parameterized problem $\Pi$ is \emph{fixed-parameter tractable} if there is an algorithm, called a \emph{fixed-parameter algorithm} for $\Pi$, that correctly confirms whether an input instance $(I,k)$ is contained in $\Pi$ in time $f(k)\cdot n^{O(1)}$ for the size $n$ of $I$ and some computable function $f:\mathbb{N}\to\mathbb{N}$.

For a parameterized problem~$\Pi$, an instance $(I,k)$ is a \emph{yes-instance} of~$\Pi$ if $(I,k)\in\Pi$, and a \emph{no-instance} of~$\Pi$, otherwise.
Instances $(I,k)$ and $(I',k')$ are \emph{equivalent} with respect to~$\Pi$ if both of them are yes-instances or are no-instances of $\Pi$.
A \emph{kernel} for~$\Pi$ is a polynomial-time preprocessing algorithm that given an instance $(I,k)$ of~$\Pi$, outputs an instance $(I',k')$ equivalent to $(I,k)$ with respect to~$\Pi$ such that $\abs{I'}+k'\leq g(k)$ for some computable function $g:\mathbb{N}\to\mathbb{N}$.
Such a function $g(k)$ is the \emph{size} of the kernel.
A \emph{polynomial kernel} for $\Pi$ is a kernel for $\Pi$ with the size as a polynomial in~$k$.
We may omit the term ``for~$\Pi$'' and ``with respect to $\Pi$'' of all these definitions if it is clear from the context.
It is known that a decidable problem~$\Pi$ is fixed-parameter tractable if and only if $\Pi$ admits a kernel; see Cygan et al.~\cite[Lemma 2.2]{CyganFKLMPPS15}.
Thus it is an interesting problem to investigate which fixed-parameter tractable problems admit a polynomial kernel.
Some problems are proved to have no polynomial kernel under some complexity hypothesis; see~Fomin et al.~\cite{kernelization}.

\subsection{Characterizations of $3$-leaf powers}\label{sec:3-leaf powers}

The graphs in Figure~\ref{fig:1} are called a \emph{bull}, a \emph{dart}, a \emph{gem}, a \emph{house}, and a \emph{domino}, respectively.
A \emph{hole} is an induced cycle of length at least $4$.
A graph is \emph{chordal} if it has no holes.
Dom, Guo, H\"{u}ffner, and Niedermeier~\cite{DGHL2006} presented the following characterization of $3$-leaf powers.

\begin{THM}[Dom, Guo, H\"{u}ffner, and Niedermeier {\cite[Theorem~1]{DGHL2006}}]\label{thm:obs 3-leaf powers}
	A graph $G$ is a $3$-leaf power if and only if $G$ is (bull, dart, gem)-free and chordal.
\end{THM}

Based on Theorem~\ref{thm:obs 3-leaf powers}, we say that a graph $H$ is an \emph{obstruction} if $H$ either is a hole, or is isomorphic to one of the bull, the dart, and the gem.
An obstruction $H$ is \emph{small} if $\abs{V(H)}\leq5$.
The following lemma observes several properties of obstructions.
We omit its straightforward proof.

\begin{LEM}\label{lem:observation}
	Every obstruction satisfies all of the following properties.
	\begin{enumerate}[label=\rm (O\arabic*),ref=O\arabic*]
		\item\label{O1} No obstructions have true twins.
		\item\label{O2} No small obstructions have an independent set of size at least $4$.
		\item\label{O3} No obstructions have $K_4$ or $K_{2,3}$ as a subgraph.
		\item\label{O4} No obstruction $H$ has a cut-vertex $v$ such that $H\setminus v$ has two components having the same number of vertices.
		\item\label{O5} False twins in an obstruction have degree $2$.
		\item\label{O6} If a vertex of an obstruction has exactly one neighbor $w$, then $w$ has degree at least $3$.
		\item\label{O7} An obstruction has at least three distinct vertices of degree $2$ if and only if it is a hole.
	\end{enumerate}
\end{LEM}

Brandst\"{a}dt and Le \cite{BL2006} showed that a graph is a $3$-leaf power if and only if it can be obtained from some forest $F$ by substituting each node $u$ of $F$ with a non-empty clique $B_u$ of arbitrary size.
We rephrase this characterization by using the following definition.

A \emph{tree-clique decomposition} of a graph $G$ is a pair $(F,\{B_u:u\in V(F)\})$ of a forest $F$ and a family $\{B_u:u\in V(F)\}$ of non-empty subsets of $V(G)$ satisfying the following two conditions.
\begin{enumerate}[label=(\arabic*)]
	\item $\{B_u:u\in V(F)\}$ is a partition of $V(G)$.
	\item \label{item:tcd2} Distinct vertices $x$ and $y$ of $G$ are adjacent if and only if $F$ has either a node $u$ such that $\{x,y\}\subseteq B_u$, or an edge $vw$ such that $x\in B_v$ and $y\in B_w$.
\end{enumerate}
For every node $u$ of $F$, we call $B_u$ a \emph{bag of $u$}.
We say that $B$ is a \emph{bag of~$G$} if $B$ is a bag of some node of~$F$
in some tree-clique decomposition $(F,\{B_u\}_{u\in V(F)})$ of~$G$.
Note that each bag is a clique by \ref{item:tcd2}.

\begin{THM}[Brandst\"{a}dt and Le {\cite[Theorem~14]{BL2006}}]\label{thm:tree-clique decomposition}
	A graph is a $3$-leaf power if and only if it has a tree-clique decomposition.
	One can construct a tree-clique decomposition of a $3$-leaf power in linear time.
	Moreover, every connected incomplete $3$-leaf power has a unique tree-clique decomposition.
\end{THM}

We remark that every connected incomplete $3$-leaf power has at least three bags.
Brandst\"{a}dt and Le \cite{BL2006} showed that in a connected incomplete $3$-leaf power, two vertices belong to the same bag if and only if they are true twins.
Thus, for such a graph $G$, a set $X\subseteq V(G)$ is a bag of $G$ if and only if~$X$ is a maximal true twin-set in $G$.

\subsection{Characterizations of distance-hereditary graphs}\label{sec:distance-hereditary}

A graph~$G$ is \emph{distance-hereditary} if for every connected induced subgraph $H$ of $G$ and vertices $v$ and $w$ of $H$, the distance between $v$ and $w$ in~$H$ is equal to the distance between $v$ and $w$ in $G$.
Bandelt and Mulder~\cite{BM1986} presented the following characterization of distance-hereditary graphs.

\begin{THM}[Bandelt and Mulder {\cite[Theorem~2]{BM1986}}]\label{thm:obs DH}
	A graph is distance-hereditary if and only if it is (gem, house, domino)-free and has no holes of length at least $5$.
\end{THM}

Since $3$-leaf powers are chordal and both the house and the domino have a hole, every $3$-leaf power is (house, domino)-free.
Therefore, by Theorems~\ref{thm:obs 3-leaf powers} and \ref{thm:obs DH}, every $3$-leaf power is distance-hereditary.

The following lemma presents a necessary condition to be distance-hereditary.
A proof of the following lemma is readily derived from the definition of a distance-hereditary graph.

\begin{LEM}\label{lem:not DH}
	Let $P$ be an induced path of a graph $G$ having length at least~$3$.
	If $G$ has a vertex~$v$ adjacent to both ends of $P$, then $G[V(P)\cup\{v\}]$ is not distance-hereditary.
\end{LEM}

\section{Overview}\label{sec:overview}
In this section, we will present the polynomial kernel for the \LPD\ problem
without showing its proof.
All proofs will be presented in later sections.

A set $S\subseteq V(G)$ is a \emph{modulator} of a graph~$G$ if $G\setminus S$ is a $3$-leaf power.
A modulator~$S$ of~$G$ is \emph{$1$-redundant} if $G\setminus(S\setminus\{v\})$ is a $3$-leaf power for every vertex $v\in S$.
Note that if $S$ is a modulator of $G$, then for every induced subgraph~$G'$ of~$G$, the set $S\cap V(G')$ is a modulator of~$G'$
and therefore if $(G,k)$ is a yes-instance for the \LPD{} problem, then so is $(G',k)$.
We remark that if $S$ is a $1$-redundant modulator of~$G$, then every obstruction in~$G$ has at least two vertices in~$S$.

The following proposition ensures that we can find a $1$-redundant modulator of small size.
Its proof will be presented in Section~\ref{sec:redundant}.

\begin{PROPredundant}
	There exists a polynomial-time algorithm that for every instance $(G,k)$ of the \LPD\ problem with $k>0$, outputs an equivalent instance $(G',k')$ and a $1$-redundant modulator $S$ of $G'$ such that $\abs{V(G')}\leq\abs{V(G)}$, $k'\leq k$, and $\abs{S}\leq84k^2+7k$.
\end{PROPredundant}

To introduce reduction rules, we need to define a few concepts.
A \emph{blocking pair} for $X\subseteq V(G)$ is an unordered pair $\{v,w\}$ of distinct vertices in $N_G(X)$ such that if $v$ and $w$ are adjacent and $N_G(v)\cap X=N_G(w)\cap X$, then $N_G(v)\cap X$ is not a clique.
We say that $X$ is \emph{blocked by $\{v,w\}$} if $\{v,w\}$ is a blocking pair for $X$.
We remark that if $N_G(X)$ has a blocking pair $\{v,w\}$ for some subset of $X$, then $X$ is also blocked by $\{v,w\}$.

For a graph $Q$, a set $X\subseteq V(Q)$, and a non-negative integer $\ell$, an \emph{$(X,\ell)$-matching} of $Q$ is a set $M\subseteq E(Q)$ such that each vertex in $X$ is incident with at most $\ell$ edges in $M$, and each vertex in $V(Q)\setminus X$ is incident with at most one edge in $M$.

We now introduce the reduction rules.
We will prove their safeness in Section~\ref{sec:safeness}.
We may assume that $k>0$,
because otherwise we can directly solve the problem by applying the linear-time algorithm of Theorem~\ref{thm:tree-clique decomposition} for recognizing $3$-leaf powers
and output $(K_1,0)$ or $(K_{2,2},0)$ depending on whether or not the input graph is a $3$-leaf power.
We may also assume that $k<\abs{V(G)}$ for the input graph~$G$,
because otherwise we can output $(K_1,0)$, as it is a yes-instance.

\begin{RULES}
	Let $(G,k)$ be an instance of the \LPD\ problem with $k>0$ and let~$S$ be a $1$-redundant modulator of $G$.
	
	\begin{enumerate}[label=(R\arabic*),ref=R\arabic*]
		\item\label{Rnumberfirst} 	Let $S^+$ be the set of vertices $u\in S$ such that for each component $C$ of~$G\setminus S$, the set $N_G(u)\cap V(C)$ is a true twin-set in $C$.
		Let $X$ be the set of $2$-element subsets of $S^+$ and $Y$ be the set of non-trivial components of $G\setminus S$ having no neighbors of $S\setminus S^+$.
		Let $Q$ be a bipartite graph on $(X\times\{1,2,3\},Y)$ satisfying all of the following conditions.
		\begin{enumerate}[label=(\arabic*)]
			\item $(\{v,w\},1)\in X\times\{1\}$ and $C\in Y$ are adjacent in $Q$ if and only if~$V(C)$ is blocked by $\{v,w\}$.
			\item $(\{v,w\},2)\in X\times\{2\}$ and $C\in Y$ are adjacent in $Q$ if and only if~$C$ has a vertex adjacent to both $v$ and $w$.
			\item $(\{v,w\},3)\in X\times\{3\}$ and $C\in Y$ are adjacent in $Q$ if and only~if~$C$ has adjacent vertices $x$ and $y$ such that both $vx$ and $wx$ are edges of $G$, and both $vy$ and $wy$ are non-edges of~$G$.
		\end{enumerate}
		If $Q$ has a maximal $(X\times\{1,2,3\},k+2)$-matching $M$ avoiding some $U\in Y$, then replace $(G,k)$ with $(G\setminus E(U),k)$.

		\item\label{Rnumbersecond}
		Let $\mathcal{A}$ be the set of ordered pairs $(A_1,A_2)$ of disjoint subsets of $S$ such that $2\leq\abs{A_1}+\abs{A_2}\leq4$, and $X$ be the set of isolated vertices in $G\setminus S$.
		For each $(A_1,A_2)\in\mathcal{A}$, let $X_{A_1,A_2}$ be a maximal set of vertices $v\in X$ such that $N_G(v)\cap(A_1\cup A_2)=A_1$ and $\abs{X_{A_1,A_2}}\leq k+3$.
		If $X\setminus\bigcup_{(A_1,A_2)\in\mathcal{A}}X_{A_1,A_2}$ contains a vertex $u$, then replace $(G,k)$ with $(G\setminus u,k)$.

		\item\label{Rcomplete}
		Let $\mathcal{A}$ be the set of ordered pairs $(A_1,A_2)$ of disjoint subsets of $S$ such that $2\leq\abs{A_1}+\abs{A_2}\leq4$, and $C$ be a complete component of $G\setminus S$.
		For each $(A_1,A_2)\in\mathcal{A}$, let $X_{A_1,A_2}$ be a maximal set of vertices~$v$ of~$C$ such that $N_G(v)\cap(A_1\cup A_2)=A_1$ and $\abs{X_{A_1,A_2}}\leq k+3$.
		If $C\setminus\bigcup_{(A_1,A_2)\in\mathcal{A}}X_{A_1,A_2}$ has a vertex~$u$, then replace $(G,k)$ with $(G\setminus u,k)$.

		\item\label{Rtruetwin}
		If $G$ has a true twin-set $X$ of size at least $k+2$, then replace $(G,k)$ with $(G\setminus v,k)$ for any vertex $v\in X$.

		\item\label{Rbagfirst}
		Let $B$ be a maximal true twin-set in $G\setminus S$.
		If $G\setminus(S\cup B)$ has a component $D$ such that $V(D)\cap N_G(S)=\emptyset$ and $V(D)\setminus N_G(B)\neq\emptyset$, then replace $(G,k)$ with $(G\setminus(V(D)\setminus N_G(B)),k)$.

		\item\label{Rbagsecond}
		Let $B$ be a maximal true twin-set in $G\setminus S$.
		If $G\setminus(S\cup B)$ has distinct components $D_1,\ldots,D_{k+4}$ such that $N_G(V(D_1))=\cdots=N_G(V(D_{k+4}))$, and either $V(D_1)\cup\cdots\cup V(D_{k+4})\subseteq N_G(B)$ or $\emptyset\neq V(D_i)\cap N_G(B)\neq V(D_i)$ for every $i\in\{1,\ldots,k+4\}$, then replace $(G,k)$ with $(G\setminus V(D_1),k)$.

		\item\label{Rbagthird}
		Let $B_1,\ldots,B_m$ be pairwise disjoint maximal true twin-sets in $G\setminus S$ for $m\geq6$ such that $N_G(B_i)=B_{i-1}\cup B_{i+1}$ for every $i\in\{2,\ldots,m-1\}$.
		Let $\ell$ be an integer in $\{3,\ldots,m-2\}$ such that $\abs{B_\ell}\leq\abs{B_i}$ for every $i\in\{3,\ldots,m-2\}$, and $G'$ be a graph obtained from $G\setminus((B_3\cup\cdots\cup B_{m-2})\setminus B_\ell)$ by making $B_\ell$ complete to $B_2\cup B_{m-1}$.
		Then replace $(G,k)$ with $(G',k)$.	
	\end{enumerate}
\end{RULES}

\begin{algorithm}[t]
	\caption{Kernelization for \LPD}
	\label{alg}
	\begin{algorithmic}[1]
		\Function {\textsc{Compress}}{$G,k$}
		\If {$k=0$}
		\If {$G$ is a $3$-leaf power}
		\Return $(K_1,0)$.
		\Else~\Return $(K_{2,2},0)$.
		\EndIf
		\EndIf
		\State {Find an instance $(G',k')$ equivalent to $(G,k)$ and a $1$-redundant modulator $S$ of $G'$ having size $O(k^2)$ by Proposition~\ref{prop:redundant}.}
		\If {$k'<k$} \Return \textsc{Compress}($G',k'$).\label{k}
		\ElsIf {$\abs{S}\leq k+1$} \Return $(K_1,0)$.
		\EndIf 
		\If {(R$i$) for some $i\in\{1,\ldots,7\}$ is applicable to $(G',k')$ with $S$} \Return \textsc{Compress}($G'',k''$) where $(G'',k'')$ is the resulting instance obtained from $(G',k')$ with $S$ by applying (R$i$).
		\Else~\Return $(G',k')$. \label{Return}
		\EndIf
		\EndFunction
	\end{algorithmic}
\end{algorithm}

Our kernel for the \LPD\ problem is presented in Algorithm~\ref{alg}.
In Section~\ref{sec:main proof}, we will show that if neither (\ref{Rnumberfirst}) nor (\ref{Rnumbersecond}) is applicable to $(G,k)$, then $G\setminus S$ has $O(k\abs{S}^2)$ non-trivial components and $O(k\abs{S}^4)$ isolated vertices.
We will also show that if none of (\ref{Rnumberfirst}), (\ref{Rcomplete}), (\ref{Rtruetwin}), (\ref{Rbagfirst}), (\ref{Rbagsecond}), and (\ref{Rbagthird}) is applicable to $(G,k)$, then each complete component of $G\setminus S$ has $O(k\abs{S}^4)$ vertices, and each incomplete component of $G\setminus S$ has $O(k^2\abs{S}^2)$ vertices.
Since the size of $S$ can be bounded from above by a polynomial in $k$ by Proposition~\ref{prop:redundant}, we will have a desired polynomial kernel.

\section{Small $1$-redundant modulators}\label{sec:redundant}

To find a small $1$-redundant modulator, we first find a modulator by combining a maximal packing of small obstructions with an outcome of an approximation algorithm for the following problem.

\medskip
\noindent
\fbox{\parbox{0.97\textwidth}{
	\textsc{\textsc{Weighted Feedback Vertex Set}}
\begin{description}
	\item[Input:] A graph $G$, a function $w:V(G)\to\mathbb{Q}\cap[0,\infty)$, and a non-negative rational number $k$
	\item[Parameter:] $k$
	\item[Question:] Is there a set $S\subseteq V(G)$ with $\sum_{v\in S}w(v)\leq k$ such that $G\setminus S$ is a forest?
\end{description}}}
\medskip

Bafna, Berman, and Fujito presented a $2$-approximation algorithm for \textsc{Weighted Feedback Vertex Set} as follows.

\begin{THM}[Bafna, Berman, and Fujito \cite{BBF1999}]\label{BBF1999}
	For a graph $G$, a function $w:V(G)\to\mathbb{Q}\cap[0,\infty)$, and a positive rational number $k$, one can either confirm that $G$ has no feedback vertex set $S$ with $\sum_{v\in S}w(v)\leq k$, or find a feedback vertex set $S$ with $\sum_{v\in S}w(v)\leq2k$ in polynomial time.
\end{THM}

Using this approximation algorithm, we can design a $7$-approximation algorithm for the \LPD\ problem.

\begin{LEM}\label{lem:modulator}
	For an instance $(G,k)$ of the \LPD\ problem with $k>0$, one can either confirm that $G$ has no modulator of size at most $k$, or find a modulator of~$G$ having at most $7k$ vertices in polynomial time.
\end{LEM}

To prove Lemma~\ref{lem:modulator}, we will use the following two lemmas.
Observe that maximal true twin-sets form a partition of the vertex set.

\begin{LEM}[Dom, Guo, H\"{u}ffner, and Niedermeier {\cite[Lemma~2]{DGHL2006}}]\label{lem:always small obs}
	Let $G$ be a graph, and $H$ be an induced subgraph of~$G$ obtained by taking exactly one vertex from each maximal true twin-set in~$G$.
	Then $H$ is triangle-free if and only if $G$ is (bull, dart, gem)-free.
\end{LEM}

\begin{LEM}\label{lem:minimal modulator}
	If a graph $G$ has a modulator $S$ and a true twin-set $X$ such that $X\setminus S$ is non-empty, then $S\setminus X$ is a modulator of $G$.
\end{LEM}
\begin{proof}
	We may assume that $S\cap X$ is non-empty.
	Suppose for contradiction that $G\setminus(S\setminus X)$ has an obstruction $H$.
	Since $S$ is a modulator of $G$, at least one vertex of $H$ is in $S$.
	Since $H$ is an induced subgraph of $G\setminus(S\setminus X)$, no vertices of $H$ is in $S\setminus X$, and therefore at least one vertex of $H$ is in $S\cap X$.
	Since $X$ is a true twin-set, by (\ref{O1}), $H$ has exactly one vertex~$v$ contained in $S\cap X$.
	Since $X\setminus S$ is non-empty, we can choose a vertex~$w$.
	Since $v$ and~$w$ are twins, $H':=G[(V(H)\setminus\{v\})\cup\{w\}]$ is isomorphic to~$H$.
	However, $H'$ is an induced subgraph of $G\setminus S$, contradicting the assumption that~$S$ is a modulator of $G$.
\end{proof}

We now prove Lemma~\ref{lem:modulator}.

\begin{proof}[Proof of Lemma~\ref{lem:modulator}]
	We can find a maximal packing $H_1,\ldots,H_m$ of vertex-disjoint small obstructions in $G$ in time $O(\abs{V(G)}^6)$.
	If $m\geq k+1$, then we confirm that $G$ has no modulators of size at most $k$.
	Thus, we may assume that $m\leq k$.
	Let $X:=\bigcup_{i\in\{1,\ldots,m\}}V(H_i)$.
	Note that $\abs{X}\leq5k$ and $G\setminus X$ has no small obstructions.
	By Theorem~\ref{thm:obs 3-leaf powers}, a set $S\subseteq V(G)$ is a modulator of $G\setminus X$ if and only if $G\setminus(S\cup X)$ is chordal.
	
	Let $H$ be an induced subgraph of $G\setminus X$ obtained by taking exactly one vertex from each maximal true twin-set in $G\setminus X$.
	Let $\sim$ be a binary relation of $V(G)\setminus X$ such that for vertices $u,v\in V(G)\setminus X$, $u\sim v$ if and only if $u=v$ or $u$ and $v$ are true twins in $G\setminus X$.
	It is easily deduced that $\sim$ is an equivalence relation.
	Thus, the set of maximal true twin-sets in $G\setminus X$ forms a partition of $V(G)\setminus X$.
	For every vertex $v\in V(G)\setminus X$, let $B_v$ be the maximal true twin-set in $G\setminus X$ containing $v$, and let $w(v):=\abs{B_v}$ for $v\in V(H)$.
	
	We claim that $(G\setminus X,k)$ is a yes-instance of the \LPD\ problem if and only if $(H,w,k)$ is a yes-instance of \textsc{Weighted Feedback Vertex Set}.
	We first show that if $G\setminus X$ has a minimal modulator~$S$ of size at most $k$, then $S\cap V(H)$ is a feedback vertex set of $H$ with $\sum_{v\in S\cap V(H)}w(v)\leq k$.
	Since~$S$ is a modulator of $G\setminus X$, the graph $G\setminus(X\cup S)$ is chordal.
	Since~$H\setminus S$ is an induced subgraph of $G\setminus(X\cup S)$, the graph~$H\setminus S$ is also chordal.
	Since $G\setminus X$ has no small obstructions, by Lemma~\ref{lem:always small obs}, $H$ is triangle-free, and so is $H\setminus S$.
	By the definition of modulators, $H\setminus S$ is hole-free.
	Therefore, $H\setminus S$ is acyclic, that is, $S\cap V(H)$ is a feedback vertex set of~$H$.
	By Lemma~\ref{lem:minimal modulator}, $B_v\subseteq S$ for every vertex $v\in S$.
	Hence, $\abs{S}=\sum_{v\in S\cap V(H)}w(v)\leq k$.
	
	Conversely, we show that if $H$ has a feedback vertex set $T$ with the sum $\sum_{v\in T}w(v)$ at most~$k$, then $B_T:=\bigcup_{v\in T}B_v$ is a modulator of $G\setminus X$ having at most~$k$ vertices.
	Note that $\abs{B_T}=\sum_{v\in T}w(v)\leq k$.
	Suppose for contradiction that $B_T$ is not a modulator of $G\setminus X$.
	Since $G\setminus X$ has no small obstruction, $G\setminus(X\cup B_T)$ has a hole $H'$ of length at least $6$.
	For each $v\in V(H')$, the set~$B_v$ is a maximal true twin-set, and therefore
	\[
		\left(\bigcup_{v\in V(H')}B_v\right)\cap V(H)
	\]
	is a hole of $H\setminus T$, contradicting the assumption that $T$ is a feedback vertex set of $H$.
	This proves the claim.
	 
	We apply Theorem~\ref{BBF1999} for $H$, $w$, and $k$.
	If the algorithm of Theorem~\ref{BBF1999} confirms that $H$ has no feedback vertex set $T$ with $\sum_{v\in T}w(v)\leq k$, then we confirm that $G$ has no modulator of size at most $k$.	
	Thus, we may assume that the algorithm of Theorem~\ref{BBF1999} outputs a set $T\subseteq V(H)$ with $\sum_{v\in T}w(v)\leq2k$.
	Note that $\abs{B_T}=\sum_{v\in T}w(v)\leq2k$.
	Then $X\cup B_T$ is a modulator of $G$ with $\abs{X\cup B_T}=\abs{X}+\abs{B_T}\leq5k+2k$.
\end{proof}

With a modulator of size $O(k)$ at hand, we are ready to find a small $1$-redundant modulator.
We note that, in principle, a small $1$-redundant modulator might not exist, but if that is the case, we are able to identify a vertex that has to be in every modulator of size at most $k$, so that we can remove it from the input graph.
To find such a vertex, we will use the following lemma, which we slightly rephrase for our application.

\begin{LEM}[Jansen and Pilipczuk {\cite[Lemma~1.3]{JP2018}}]\label{lem:EPflower}
	For a graph $G$, a non-negative integer $k$, and a vertex $v$ of $G$, if $G\setminus v$ is chordal, then one can in polynomial time either find  holes $H_1,\ldots,H_{k+1}$ in $G$ such that $V(H_i)\cap V(H_j)=\{v\}$ for all distinct $i,j\in\{1,\ldots,k+1\}$, or find a set $S\subseteq V(G)\setminus\{v\}$ of size at most $12k$ such that $G\setminus S$ is chordal.
\end{LEM}

We now present an algorithm for finding a $1$-redundant modulator of small size.

\begin{PROP}\label{prop:redundant}
	There exists a polynomial-time algorithm that for every instance $(G,k)$ of the \LPD\ problem with $k>0$, outputs an equivalent instance $(G',k')$ and a $1$-redundant modulator $S$ of $G'$ such that $\abs{V(G')}\leq\abs{V(G)}$, $k'\leq k$, and $\abs{S}\leq84k^2+7k$.
\end{PROP}
\begin{proof}
	We first try to find a modulator $S$ of $G$ having size at most $7k$ by using Lemma~\ref{lem:modulator}.
	If it fails, then $(G,k)$ is a no-instance, and therefore we output $(K_{2,2},0)$ as $(G',k')$ and $V(K_{2,2})$ as a $1$-redundant modulator of $G'$.
	Otherwise, for each vertex $v$ in $S$, let $G_v:=G\setminus(S\setminus\{v\})$.
	Note that every obstruction in~$G_v$ contains $v$.
	Let $F^v_1,\ldots,F^v_{m(v)}$ be a maximal packing of small obstructions in $G_v$ such that $V(F^v_i)\cap V(F^v_j)=\{v\}$ for all distinct $i,j\in\{1,\ldots,m(v)\}$.
	We now let
	\[
		G'_v:=G_v\setminus((V(F^v_1)\cup\cdots\cup V(F^v_{m(v)}))\setminus\{v\}).
	\]
	
	If $m(v)\geq k+1$ for some vertex $v\in S$, then every modulator of $G$ having size at most $k$ must contain $v$, because otherwise it contains at least one vertex in $V(F^v_i)\setminus\{v\}$ for each $i\in\{1,\ldots,m(v)\}$, which are pairwise disjoint, contradicting that the modulator has size at most $k$.
	Thus, in this case, $(G,k)$ is equivalent to $(G\setminus v,k-1)$, so that we can apply our algorithm recursively for $(G\setminus v,k-1)$.
	Therefore, we may assume that $m(v)\leq k$ for every vertex $v\in S$.
	
	By applying Lemma~\ref{lem:EPflower} for $G'_v$, $k-m(v)$, and $v$, we can either
	\begin{enumerate}
		\item [(i)] find $k-m(v)+1$ holes $H^v_1,\ldots,H^v_{k-m(v)+1}$ in $G'_v$ such that $V(H^v_i)\cap V(H^v_j)=\{v\}$ for all distinct $i,j\in\{1,\ldots,k-m(v)+1\}$, or
		\item [(ii)] find a set $S'_v\subseteq V(G'_v)\setminus\{v\}$ of size at most $12(k-m(v))$ such that $G'_v\setminus S'_v$ is chordal.
	\end{enumerate}
	
	If (i) holds, then we apply our algorithm recursively for $(G\setminus v,k-1)$,
	because if a modulator does not contain $v$, then 
	it contains at least one vertex in $V(F^v_i)\setminus\{v\}$ for each $i\in\{1,\ldots,m(v)\}$
	and at least one vertex in $V(H^v_j)\setminus\{v\}$ for each $j\in\{1,\ldots,k-m(v)+1\}$.
	Therefore, we may assume that (ii) holds for every vertex $v\in S$.
	Let
	\[
		S_v:=(V(F^v_1)\cup\cdots\cup V(F^v_{m(v)})\cup S'_v)\setminus\{v\}.
	\]
	Note that $\abs{S_v}\leq4m(v)+12(k-m(v))\le 12k$ and $G_v\setminus S_v$ is a $3$-leaf power.
	
	We output $(G,k)$ as $(G',k')$ and $X:=S\cup \bigcup_{v\in S}S_v$ as a $1$-redundant modulator of $G$.
	Clearly, $\abs{X}\leq\abs{S}+12k\abs{S}\leq84k^2+7k$.
	It remains to argue that~$X$ is indeed a $1$-redundant modulator of $G$.
	Let $H$ be an arbitrary obstruction in $G$.
	It suffices to show that $\abs{V(H)\cap X}\geq2$.
	Since~$S$ is a modulator of $G$, we have that $H$ has a vertex $v\in S$.
	If $\abs{V(H)\cap S}=1$, then~$H$ is an induced subgraph of $G_v$, and therefore $H$ has at least one vertex in~$S_v$.
	Since $S_v$ and $S$ are disjoint,  $\abs{V(H)\cap X}\geq2$, and therefore $X$ is a $1$-redundant modulator of~$G$.
\end{proof}

\section{Safeness of the reduction rules}\label{sec:safeness}

We now aim to show the safeness of our reduction rules.
In Subsection~\ref{subsec:technical}, we will introduce the notion of \emph{complete split} of a graph and present four technical lemmas.
In Subsection~\ref{subsec:safeness12}, we will show the safeness of (\ref{Rnumberfirst}), \ldots, (\ref{Rbagfirst}).
In Subsections~\ref{subsec:safeness6} and~\ref{subsec:safeness7}, we will show the safeness of (\ref{Rbagsecond}) and (\ref{Rbagthird}), respectively.

\subsection{Complete splits and blocking pairs}\label{subsec:technical}

Cunningham \cite{Cunningham1982} introduced a split of a graph.
A \emph{split} of a graph $G$ is a partition $(A,B)$ of $V(G)$ such that $\min\{\abs{A},\abs{B}\}\geq2$ and $N(A)$ is complete to $N(B)$.
We say that a split $(A,B)$ of $G$ is \emph{complete} if $N(A)\cup N(B)$ is a clique.
If a graph has a complete split, then obstructions must satisfy some conditions which we prove in the following two lemmas.

\begin{LEM}\label{lem:no holes between}
	Let $G$ be a graph and $(A,B)$ be a complete split of $G$.
	If $G$ has a hole~$H$, then $V(H)\cap A=\emptyset$ or $V(H)\cap B=\emptyset$.
\end{LEM}
\begin{proof}
	Suppose not.
	Since $N(A)\cup N(B)$ is a clique, $H$ has at most two vertices in $N(A)\cup N(B)$, because otherwise $H$ has $K_3$ as a subgraph.
	Since both $V(H)\cap A$ and $V(H)\cap B$ are non-empty and $H$ is connected, $H$ has vertices $x_1\in N(A)$ and $x_2\in N(B)$.
	Thus, $H$ has exactly two vertices $x_1$ and~$x_2$ in $N(A)\cup N(B)$, so that $H\setminus x_1x_2$ is disconnected, a contradiction.
\end{proof}

\begin{LEM}\label{lem:obs between}
	Let $G$ be a graph and $(A,B)$ be a complete split of $G$.
	If $G$ has an obstruction $H$ having exactly two vertices in $A$, then $H$ is isomorphic to the bull.
\end{LEM}
\begin{proof}
	We denote by~$a_1$ and~$a_2$ the vertices in $V(H)\cap A$.
	Note that~$H$ has at least three vertices in~$B$.
	Suppose for contradiction that both $a_1$ and $a_2$ have neighbors in~$B$.
	Since $N(A)\cup N(B)$ is a clique, $a_1$ and $a_2$ are adjacent and have the same set of neighbors in~$B$.
	Then $a_1$ and $a_2$ are true twins in~$H$, contradicting (\ref{O1}).
	
	Therefore, either $a_1$ or $a_2$ has no neighbors in $B$.
	By symmetry, we may assume that $a_1$ has no neighbors in $B$.
	Since $H$ is connected, $a_1$ is adjacent to $a_2$.
	Thus, $a_1$ has degree $1$ in~$H$.
	Since $N(A)\cup N(B)$ is a clique, by~(\ref{O3}),~$a_2$ has at most three neighbors in~$H$.
	Then by (\ref{O6}), $a_2$ has degree~$3$ in~$H$.
	Hence, the bull is the only possible obstruction for $H$.
\end{proof}

Recall that a \emph{blocking pair} for $X\subseteq V(G)$ is an unordered pair $\{v,w\}$ of distinct vertices in $N(X)$ such that if $v$ and $w$ are adjacent and $N(v)\cap X=N(w)\cap X$, then $N(v)\cap X$ is not a clique.
The following lemma motivates this definition.

\begin{LEM}\label{lem:blocking pair}
	Let $G$ be a graph and $(A,B)$ be a partition of $V(G)$ such that $\min\{\abs{A},\abs{B}\}\geq2$.
	Then $(A,B)$ is a complete split of $G$ if and only if $N(B)$ is a clique and $B$ has no blocking pairs for $A$.
\end{LEM}
\begin{proof}
	It is clear that if $(A,B)$ is a complete split of $G$, then $N(B)$ is a clique and $B$ has no blocking pairs for $A$.
	
	Conversely, suppose that $N(B)$ is a clique and $B$ has no blocking pairs for $A$.
	We may assume that $\abs{N(A)}\geq2$, because otherwise $N(A)\cup N(B)$ is a clique and $(A,B)$ is a complete split of $G$.
	Since $B$ has no blocking pairs for $A$, we have that $N(A)$ is a clique, because if $N(A)$ has a non-edge $vw$, then $\{v,w\}$ is a blocking pair for $A$.
	Moreover, $N(v)\cap A=N(w)\cap A$ for each pair of distinct vertices $v$ and $w$ in $N(A)$, because otherwise $\{v,w\}$ is a blocking pair for $A$.
	This means that $N(A)$ is complete to $N(B)$.
	Thus, $N(A)\cup N(B)$ is a clique, and therefore $(A,B)$ is a complete split of~$G$.
\end{proof}

The following lemma shows that if there is a blocking pair $\{v,w\}$ for a set $X\subseteq V(G)$ such that $G[X]$ has two distinct components whose vertex sets are blocked by $\{v,w\}$, then $G$ is not a $3$-leaf power.
\begin{LEM}\label{lem:obs from blocking}
	Let $G$ be a graph and $(A,B)$ be a partition of $V(G)$ such that $\min\{\abs{A},\abs{B}\}\geq2$.
	If $G[A]$ has distinct components $C_1$ and $C_2$ such that both $V(C_1)$ and $V(C_2)$ are blocked by a pair $\{v,w\}$ of vertices in $B$, then $G[V(C_1)\cup V(C_2)\cup\{v,w\}]$ is not a $3$-leaf power.
\end{LEM}
\begin{proof}
	Suppose first that $N_G(v)\cap V(C_i)\neq N_G(w)\cap V(C_i)$ for some $i\in\{1,2\}$.
	By symmetry, we may assume that $N_G(v)\cap V(C_1)\neq N_G(w)\cap V(C_1)$.
	Then~$C_1$ has a vertex $u_1$ adjacent to exactly one of $v$ and $w$.
	By symmetry, we may assume that $u_1$ is adjacent to $v$ and non-adjacent to $w$.
	Since $\{v,w\}$ is a blocking pair for $V(C_2)$, there exist a neighbor $u_2$ of $v$ in $C_2$.
	Since $G[V(C_1)\cup V(C_2)\cup\{w\}]$ is connected, we can find an induced path $P$ in the graph between~$u_1$ and~$u_2$.
	Note that the length of $P$ is at least $3$, because~$P$ must intersect $w$ that is non-adjacent to~$u_1$.
	Since $v$ is adjacent to both ends of $P$, by Lemma~\ref{lem:not DH}, $G[V(P)\cup\{v\}]$ is not distance-hereditary.
	Therefore, $G[V(C_1)\cup V(C_2)\cup\{v,w\}]$ is not a $3$-leaf power.
	
	We now suppose that $N_G(v)\cap V(C_i)=N_G(w)\cap V(C_i)$ for $i=1,2$.
	If~$v$ and~$w$ are non-adjacent, then for neighbors $u_1,u_2$ of $v$ with $u_1\in V(C_1)$ and $u_2\in V(C_2)$, the graph $G[\{v,w,u_1,u_2\}]$ is a hole.
	Thus, we may assume that~$v$ and~$w$ are adjacent.
	Since $\{v,w\}$ is a blocking pair for $V(C_1)$ and $N_G(v)\cap V(C_1)=N_G(w)\cap V(C_1)$, there exists a non-edge $u_1u_2$ in $N_G(v)\cap V(C_1)$.
	Let $P$ be an induced path in $C_1$ between $u_1$ and $u_2$.
	Since $v$ is adjacent to both ends of $P$, by Lemma~\ref{lem:not DH}, we may assume that the length of $P$ is exactly $2$.
	Let $u_3$ be the internal vertex of $P$, and $u_4$ be a neighbor of~$v$ in $V(C_2)$.
	Then $G[\{v,u_1,u_2,u_3,u_4\}]$ is isomorphic to the dart if $u_3$ is adjacent to $v$, and has a hole of length $4$ if $u_3$ is non-adjacent to $v$.
	Therefore, $G[V(C_1)\cup V(C_2)\cup\{v,w\}]$ is not a $3$-leaf power.
\end{proof}

\subsection{(\ref{Rnumberfirst}), \ldots, (\ref{Rbagfirst}) are safe}\label{subsec:safeness12}

We first show the safeness of (\ref{Rnumberfirst}).

\begin{RULE}[\ref{Rnumberfirst}]
	Let $S^+$ be the set of vertices $u\in S$ such that for each component $C$ of~$G\setminus S$, the set $N_G(u)\cap V(C)$ is a true twin-set in $C$.
	Let $X$ be the set of $2$-element subsets of~$S^+$ and~$Y$ be the set of non-trivial components of $G\setminus S$ having no neighbors of $S\setminus S^+$.
	Let $Q$ be a bipartite graph on $(X\times\{1,2,3\},Y)$ satisfying all of the following conditions.
	\begin{enumerate}[label=(\arabic*)]
		\item $(\{v,w\},1)\in X\times\{1\}$ and $C\in Y$ are adjacent in $Q$ if and only if~$V(C)$ is blocked by $\{v,w\}$.
		\item $(\{v,w\},2)\in X\times\{2\}$ and $C\in Y$ are adjacent in $Q$ if and only if~$C$ has a vertex adjacent to both $v$ and $w$.
		\item $(\{v,w\},3)\in X\times\{3\}$ and $C\in Y$ are adjacent in $Q$ if and only if~$C$ has adjacent vertices $x$ and $y$ such that both $vx$ and $wx$ are edges of $G$, and both $vy$ and $wy$ are non-edges of $G$.
	\end{enumerate}
	If $Q$ has a maximal $(X\times\{1,2,3\},k+2)$-matching $M$ avoiding some $U\in Y$, then replace $(G,k)$ with $(G\setminus E(U),k)$.
\end{RULE}
\begin{proof}[Proof of Safeness]
	Let $G':=G\setminus E(U)$.
	We first show that if $(G,k)$ is a yes-instance, then so is $(G',k)$.
	Suppose that $G$ has a modulator $S'$ of size at most $k$ and $G'\setminus S'$ has an obstruction $H$.
	Since $G\setminus S'$ is a $3$-leaf power, $H$ has vertices $b_1$ and $b_2$ with $b_1b_2\in E(U\setminus S')$.
	Thus, $\abs{V(U)\setminus S'}\geq2$.
	
	\begin{CLM}\label{clm:Rnumberfirst-1}
		$(V(U)\setminus S',V(G)\setminus(V(U)\cup S'))$ is a split of $G'\setminus S'$.
	\end{CLM}
	\begin{subproof}[Proof of Claim~\ref{clm:Rnumberfirst-1}]
		We first show that $\abs{V(G)\setminus(V(U)\cup S')}\geq2$.
		Since $V(U)\setminus S'$ is an independent set of $G'\setminus S'$, if $H$ is a hole, then $H$ has at most $\lfloor\abs{V(H)}/2\rfloor$ vertices in $V(U)\setminus S'$, and therefore $H$ has at least $\lceil\abs{V(H)}/2\rceil\geq2$ vertices in $V(G)\setminus(V(U)\cup S')$.
		Otherwise, $H$ is a small obstruction, so by (\ref{O2}), $H$ has at most three vertices in $V(U)\setminus S'$, and therefore $H$ has at least two vertices in $V(G)\setminus(V(U)\cup S')$.
		Hence, $\abs{V(G)\setminus(V(U)\cup S')}\geq2$.
		
		Suppose for contradiction that $(V(U)\setminus S',V(G)\setminus(V(U)\cup S'))$ is not a split of $G'\setminus S'$.
		Then $V(G)\setminus(V(U)\cup S')$ contains vertices $v$ and $w$ such that both $v$ and $w$ have neighbors in $V(U)\setminus S'$ and $N_G(v)\cap(V(U)\setminus S')\neq N_G(w)\cap(V(U)\setminus S')$.
		Thus, $\{v,w\}$ is a blocking pair for $V(U)\setminus S'$, so for $V(U)$.
		Then $U$ is adjacent to $(\{v,w\},1)$ in $Q$.
		Since $M$ is maximal, $Y$ has distinct elements $C_1,\ldots,C_{k+2}$ different from $U$ such that $V(C_i)$ is blocked by $\{v,w\}$ for every $i\in\{1,\ldots,k+2\}$.
		Since $\abs{S'}\leq k$, two of them, say $C_1$ and~$C_2$, have no vertices in $S'$.
		Then $G[V(C_1)\cup V(C_2)\cup\{v,w\}]$ is not a $3$-leaf power by Lemma~\ref{lem:obs from blocking}, which is an induced subgraph of $G\setminus S'$, contradicting the assumption that $S'$ is a modulator of $G$.
	\end{subproof}
	
	Since $H$ is connected and $V(U)\setminus S'$ is an independent set of $G'\setminus S'$, both~$b_1$ and~$b_2$ have neighbors in $V(G)\setminus(V(U)\cup S')$.
	Then by Claim~\ref{clm:Rnumberfirst-1}, $b_1$ and~$b_2$ are false twins in $G'\setminus S'$.
	By (\ref{O5}), both~$b_1$ and~$b_2$ have degree $2$ in~$H$.
	Let $z_1$ and $z_2$ be the neighbors of $b_1$ in $V(H)\cap S$.
	Then $U$ is adjacent to $(\{z_1,z_2\},2)$ in $Q$.
	Since $M$ is maximal, $Y$ has distinct elements $C'_1,\ldots,C'_{k+2}$ different from $U$ such that $C'_i$ has a vertex adjacent to both $z_1$ and $z_2$ for every $i\in\{1,\ldots,k+2\}$.
	Since $\abs{S'}\leq k$, two of them, say $C'_1$ and $C'_2$, have no vertices in $S'$.
	Note that $S'$ has no vertices of $H$, because $H$ is an induced subgraph of $G'\setminus S'$.
	
	If $z_1$ and $z_2$ are non-adjacent, then $G[V(C'_1)\cup V(C'_2)\cup\{z_1,z_2\}]$ has a hole of length $4$, which is an induced subgraph of $G\setminus S'$, contradicting the assumption that $S'$ is a modulator of $G$.
	Therefore, $z_1$ and $z_2$ are adjacent.
	Since $G[\{b_1,z_1,z_2\}]$ is a triangle, $H$ is not a hole, and therefore $\abs{V(H)}=5$.
	Let $a$ be a vertex of $H$ different from $b_1,b_2,z_1$, and $z_2$.
	We may assume that $a$ is not in $V(C'_1)$, because otherwise we may swap $C'_1$ and $C'_2$.
	Let $c$ be a vertex of $C'_1$ adjacent to both $z_1$ and $z_2$.
	Note that $G[\{b_1,b_2,z_1,z_2\}]$ is isomorphic to $K_4\setminus b_1b_2$.
	Since the dart and a hole of length $4$ are the only obstructions having false twins, $H$ is isomorphic to the dart.
	Thus, $N_H(a)=\{z_i\}$ for some $i\in\{1,2\}$.
	Then $G[\{a,b_1,c,z_1,z_2\}]$ is isomorphic to the gem if $c$ is adjacent to $a$, and the dart if $c$ is non-adjacent to $a$, contradicting the assumption that $S'$ is a modulator of $G$.
	Therefore, if $(G,k)$ is a yes-instance, then so is $(G',k)$.	
	
	We now show that if $(G',k)$ is a yes-instance, then so is $(G,k)$.
	Suppose that $G'$ has a modulator~$S'$ of size at most $k$ and $G\setminus S'$ has an obstruction $H$.
	Since $G'\setminus S'$ is a $3$-leaf power, $H$ has an edge of $U\setminus S'$.
	Thus, $\abs{V(U)\setminus S'}\geq2$.
	Since $S$ is a $1$-redundant modulator of $G$, $H$ has at least two vertices in $S\setminus S'\subseteq V(G)\setminus(V(U)\cup S')$.
	Thus, $\abs{V(G)\setminus(V(U)\cup S')}\geq2$.
	
	\begin{CLM}\label{clm:Rnumberfirst-2}
		$(V(U)\setminus S',V(G)\setminus(V(U)\cup S'))$ is a complete split of $G\setminus S'$.
	\end{CLM}
	\begin{subproof}[Proof of Claim~\ref{clm:Rnumberfirst-2}]
		Suppose not.
		We first show that $V(G)\setminus(V(U)\cup S')$ has a blocking pair for $V(U)\setminus S'$.
		Since $U$ is a component of $G\setminus S$, and has no neighbors of $S\setminus S^+$, it suffices to show that $S^+\setminus S'$ has a blocking pair for $V(U)\setminus S'$.
		We may assume that for all distinct vertices $v,w\in S^+\setminus S'$, if both~$v$ and~$w$ have neighbors in $V(U)\setminus S'$, then $v$ and $w$ are adjacent and have the same neighborhoods in $V(U)\setminus S'$, because otherwise $\{v,w\}$ is a blocking pair for $V(U)\setminus S'$.
		For each vertex $v\in S^+\setminus S'$ having neighbors in $V(U)\setminus S'$, the set of neighbors of $v$ in $V(U)\setminus S'$ is a true twin-set in $U\setminus S'$, that is, a clique.
		Therefore, $N_G(S^+\setminus S')\cap(V(U)\setminus S')$ is a clique of $U\setminus S'$.
		Thus, by Lemma~\ref{lem:blocking pair}, $S^+\setminus S'$ has a blocking pair $\{v,w\}$ for $V(U)\setminus S'$, so for $V(U)$.
		
		Since $V(U)$ is blocked by $\{v,w\}$, $U$ is adjacent to $(\{v,w\},1)$ in $Q$.
		Since~$M$ is maximal, $Y$ has distinct elements $C_1,\ldots,C_{k+2}$ different from~$U$ such that $V(C_i)$ is blocked by $\{v,w\}$ for every $i\in\{1,\ldots,k+2\}$.
		Since $\abs{S'}\leq k$, two of them, say $C_1$ and $C_2$, have no vertices in $S'$.
		Then $G[V(C_1)\cup V(C_2)\cup\{v,w\}]$ is not a $3$-leaf power by Lemma~\ref{lem:obs from blocking}, which is an induced subgraph of $G'\setminus S'$, contradicting the assumption that $S'$ is a modulator of $G'$.
	\end{subproof}
	
	Since both $V(U)\setminus S'$ and $V(G)\setminus(V(U)\cup S')$ contain vertices of $H$, it is not a hole by Lemma~\ref{lem:no holes between} and Claim~\ref{clm:Rnumberfirst-2}, and therefore $\abs{V(H)}=5$.
	Let $t_1,\ldots,t_p$ be the vertices of $H$ in $V(U)\setminus S'$, and $s_1,\ldots,s_q$ be the vertices of $H$ in $V(G)\setminus(V(U)\cup S')$.
	Note that both $p$ and $q$ are at least $2$.
	Since $\abs{V(H)}=5$, either $(p,q)=(3,2)$ or $(p,q)=(2,3)$ holds.
	
	If $(p,q)=(3,2)$, then by Lemma~\ref{lem:obs between} and Claim~\ref{clm:Rnumberfirst-2}, we may assume that $N_H(s_1)=\{s_2\}$ and $N_H(s_2)=\{s_1,t_1,t_2\}$.
	Since $U$ has no neighbors of $S\setminus S^+$, we have that $s_2\in S^+$, so $t_1$ and $t_2$ are true twins in $U\setminus S'$, contradicting (\ref{O1}).
	
	Therefore, $(p,q)=(2,3)$.
	By Lemma~\ref{lem:obs between} and Claim~\ref{clm:Rnumberfirst-2}, we may assume that $N_H(t_1)=\{t_2\}$ and $N_H(t_2)=\{t_1,s_1,s_2\}$.
	Note that $\{s_1,s_2\}\subseteq S\setminus S'$.
	Then $U$ is adjacent to $(\{s_1,s_2\},3)$ in $Q$.
	Since $M$ is maximal, $Y$ has distinct elements $C''_1,\ldots,C''_{k+2}$ different from $U$ such that $C''_i$ has an edge $x_iy_i$ such that $x_i$ is adjacent to both~$s_1$ and~$s_2$, and $y_i$ is non-adjacent to both~$s_1$ and~$s_2$ for every $i\in\{1,\ldots,k+2\}$.
	Since $\abs{S'}\leq k$, two of them, say $C''_1$ and~$C''_2$, have no vertices in $S'$.
	We may assume that $s_3$ is not in $V(C''_1)$, because otherwise we may swap $C''_1$ and $C''_2$.
	We remark that the bull is the only possible graph to which $H$ is isomorphic.
	Thus, $s_1$ and $s_2$ are adjacent, and~$s_3$ is adjacent to exactly one of~$s_1$ and~$s_2$ in~$H$.
	Then $G[\{x_1,y_1,s_1,s_2,s_3\}]$ is isomorphic to the gem if both~$x_1$ and~$y_1$ are adjacent to~$s_3$, the bull if both~$x_1$ and~$y_1$ are non-adjacent to $s_3$, and the dart if $x_1$ is adjacent to $s_3$ and $y_1$ is non-adjacent to $s_3$, and has a hole of length $4$ if $x_1$ is non-adjacent to~$s_3$ and~$y_1$ is adjacent to~$s_3$.
	However, the graph is an induced subgraph of $G'\setminus S'$, contradicting the assumption that $S'$ is a modulator of $G'$.
	Therefore, if $(G',k)$ is a yes-instance, then so is $(G,k)$.
\end{proof}

Next we prove that (\ref{Rnumbersecond}) is safe.

\begin{RULE}[\ref{Rnumbersecond}]
	Let $\mathcal{A}$ be the set of ordered pairs $(A_1,A_2)$ of disjoint subsets of $S$ such that $2\leq\abs{A_1}+\abs{A_2}\leq4$, and $X$ be the set of isolated vertices in $G\setminus S$.
	For each $(A_1,A_2)\in\mathcal{A}$, let $X_{A_1,A_2}$ be a maximal set of vertices $v\in X$ such that $N_G(v)\cap(A_1\cup A_2)=A_1$ and $\abs{X_{A_1,A_2}}\leq k+3$.
	If $X\setminus\bigcup_{(A_1,A_2)\in\mathcal{A}}X_{A_1,A_2}$ contains a vertex $u$, then replace $(G,k)$ with $(G\setminus u,k)$.
\end{RULE}
\begin{proof}[Proof of Safeness]
	If $(G,k)$ is a yes-instance, then there is a modulator $S'$ of size at most $k$.
	Then $(G\setminus u) \setminus (S'\setminus\{u\})$ is also a $3$-leaf power, implying that $(G\setminus u,k)$ is a yes-instance.
	
	We now show that if $(G\setminus u,k)$ is a yes-instance, then so is $(G,k)$.
	Suppose that $G\setminus u$ has a modulator~$S'$ of size at most $k$ and $G\setminus S'$ has an obstruction~$H$.
	Since $G\setminus(S'\cup\{u\})$ is a $3$-leaf power, $H$ has~$u$.
	
	Suppose first that $H$ is a hole.
	Then $u$ has exactly two neighbors $v_1,v_2\in S\cap V(H)$ which are non-adjacent.
	By definition, $X_{\{v_1,v_2\},\emptyset}$ contains distinct vertices $u_1,\ldots,u_{k+3}$ different from~$u$.
	Note that $H$ has at most one of $u_1,\ldots,u_{k+3}$, because $v_1$ and $v_2$ have at most one common neighbor in $V(H)$ except for~$u$.
	Since $\abs{S'}\leq k$, two of them, say $u_1$ and $u_2$, are not in $S'\cup V(H)$.
	Thus, $G[\{v_1,v_2,u_1,u_2\}]$ is a hole, which is an induced subgraph of $G\setminus(S'\cup\{u\})$, contradicting the assumption that $S'$ is a modulator of $G\setminus u$.
	
	Therefore, $\abs{V(H)}=5$.
	Note that $2\leq\abs{S\cap V(H)}\leq4$, because $S$ is a $1$-redundant modulator of $G$ and $u\in V(H)\setminus S$.
	Let $B_1:=(S\cap V(H))\cap N_G(u)$, and $B_2:=(S\cap V(H))\setminus N_G(u)$.
	By definition, $X_{B_1,B_2}$ contains distinct vertices $u_1,\ldots,u_{k+3}$ different from $u$.
	Since $\abs{V(H)}=5$ and $2\leq\abs{S\cap V(H)}\leq4$, $H$ has at most three vertices in $X$ including $u$.
	Thus, $H$ has at most two of $u_1,\ldots,u_{k+3}$.
	Since $\abs{S'}\leq k$, one of them, say $u_1$, is not in $S'\cup V(H)$.
	Thus, $G[(V(H)\setminus\{u\})\cup\{u_1\}]$ is isomorphic to $H$, which is an induced subgraph of $G\setminus(S'\cup\{u\})$, contradicting the assumption that~$S'$ is a modulator of $G\setminus u$.
	Therefore, if $(G\setminus u,k)$ is a yes-instance, then so is~$(G,k)$.
\end{proof}

Here is the proof that (\ref{Rcomplete}) is safe.

\begin{RULE}[\ref{Rcomplete}]
	Let $\mathcal{A}$ be the set of ordered pairs $(A_1,A_2)$ of disjoint subsets of $S$ such that $2\leq\abs{A_1}+\abs{A_2}\leq4$, and $C$ be a complete component of $G\setminus S$.
	For each $(A_1,A_2)\in\mathcal{A}$, let $X_{A_1,A_2}$ be a maximal set of vertices~$v$ of~$C$ such that $N_G(v)\cap(A_1\cup A_2)=A_1$ and $\abs{X_{A_1,A_2}}\leq k+3$.
	If $C\setminus\bigcup_{(A_1,A_2)\in\mathcal{A}}X_{A_1,A_2}$ has a vertex~$u$, then replace $(G,k)$ with $(G\setminus u,k)$.
\end{RULE}
\begin{proof}[Proof of Safeness]
	If $(G,k)$ is a yes-instance, then there is a modulator $S'$ of size at most $k$.
	Then $(G\setminus u)\setminus(S'\setminus\{u\})$ is also a $3$-leaf power, implying that $(G\setminus u,k)$ is a yes-instance.

	We now show that if $(G\setminus u,k)$ is a yes-instance, then so is $(G,k)$.
	Suppose that $G\setminus u$ has a modulator~$S'$ of size at most $k$, and $G\setminus S'$ has an obstruction~$H$.
	Since $G\setminus(S'\cup\{u\})$ is a $3$-leaf power, $H$ has $u$.
	
	Suppose first that $H$ is a small obstruction.
	Note that $2\leq\abs{S\cap V(H)}\leq~4$, because $S$ is a $1$-redundant modulator of $G$ and $u\in V(H)\setminus S$.
	Let $B_1:=(S\cap V(H))\cap N_G(u)$, and $B_2:=(S\cap V(H))\setminus N_G(u)$.
	By definition, $X_{B_1,B_2}$ contains distinct vertices $u_1,\ldots,u_{k+3}$ different from~$u$.
	Since $\abs{V(H)}\leq5$ and $2\leq\abs{S\cap V(H)}\leq4$, $H$ has at most three vertices of $C$ including~$u$.
	Thus,~$H$ has at most two of $u_1,\ldots,u_{k+3}$.
	Since $\abs{S'}\leq k$, one of them, say~$u_1$, is not in $S'\cup V(H)$.
	Thus, $G[(V(H)\setminus\{u\})\cup\{u_1\}]$ is isomorphic to $H$, which is an induced subgraph of $G\setminus(S'\cup\{u\})$, contradicting the assumption that~$S'$ is a modulator of $G\setminus u$.
	
	Therefore, $H$ is a hole of length at least 6.
	Since $C$ is complete, $H$ has at most two vertices of $C$.
	We first consider the case that $H$ has exactly one vertex $u$ of $C$.
	In this case, $V(H)\cap S$ contains distinct vertices~$v_1$ and~$v_2$ which are neighbors of $u$.
	Then $H\setminus u$ is an induced path of length at least~$4$ between~$v_1$ and~$v_2$.
	By definition, $X_{\{v_1,v_2\},\emptyset}$ contains distinct vertices $u_1,\ldots,u_{k+3}$ different from $u$.
	Since $\abs{S'}\leq k$, one of them, say $u_1$, is not in~$S'$.
	Then $G[(V(H)\setminus\{u\})\cup\{u_1\}]$ is not distance-hereditary by Lemma~\ref{lem:not DH}, which is an induced subgraph of $G\setminus(S'\cup\{u\})$, contradicting the assumption that $S'$ is a modulator of $G\setminus u$.
	
	We now consider that $H$ has exactly two vertices $u$ and $u'$ of $C$.
	In this case, $V(H)\cap S$ contains distinct vertices $v_1$ and~$v_2$such that $v_1$ is adjacent to $u$ and $v_2$ is adjacent to $u'$.
	Note that $u'$ is non-adjacent to $v_1$.
	Then $H\setminus u$ is an induced path of length at least $4$ between $v_1$ and $u'$.
	By definition, $X_{\{v_1\},\{v_2\}}$ contains distinct vertices $u_1,\ldots,u_{k+3}$ different from $u$.
	Since $\abs{S'}\leq k$, one of them, say $u_1$, is not in $S'$.
	Then $G[(V(H)\setminus\{u\})\cup\{u_1\}]$ is not distance-hereditary by Lemma~\ref{lem:not DH}, which is an induced subgraph of $G\setminus(S'\cup\{u\})$, contradicting the assumption that $S'$ is a modulator of $G\setminus u$.
	Therefore, if $(G\setminus u,k)$ is a yes-instance, then so is~$(G,k)$.
\end{proof}

We now show that one can avoid large true twin-sets.

\begin{RULE}[\ref{Rtruetwin}]
	If $G$ has a true twin-set $X$ of size at least $k+2$, then replace $(G,k)$ with $(G\setminus v,k)$ for any vertex $v\in X$.
\end{RULE}
\begin{proof}[Proof of Safeness]
	If $(G,k)$ is a yes-instance, then there is a modulator $S'$ of size at most $k$.
	Then $(G\setminus v)\setminus(S'\setminus\{v\})$ is also a $3$-leaf power, implying that $(G\setminus v,k)$ is a yes-instance.
	
	We now show that if $(G\setminus v,k)$ is a yes-instance, then so is $(G,k)$.
	Suppose that $G\setminus v$ has a modulator~$S'$ of size at most $k$ and~$G\setminus S'$ has an obstruction~$H$.
	Since $G\setminus(S'\cup\{v\})$ is a $3$-leaf power, $H$ has $v$.
	By~(\ref{O1}),~$v$ is the only vertex in $V(H)\cap X$.
	Since $\abs{S'}\leq k$, $X$ contains a vertex~$w$ not in $S'\cup\{v\}$.
	Then $G[(V(H)\setminus\{v\})\cup\{w\}]$ is isomorphic to~$H$, which is an induced subgraph of $G\setminus(S'\cup\{v\})$, contradicting that $S'$ is a modulator of $G\setminus v$.
\end{proof}

We now show that if 
$B$ is a bag corresponding to a leaf in a tree-clique decomposition of $G\setminus S$ for a $1$-redundant modulator $S$ of a graph $G$, then we may assume that it has a neighbor in~$S$.

\begin{RULE}[\ref{Rbagfirst}]
	Let $B$ be a maximal true twin-set in $G\setminus S$.
	If $G\setminus(S\cup B)$ has a component $D$ such that $V(D)\cap N_G(S)=\emptyset$ and $V(D)\setminus N_G(B)\neq\emptyset$, then replace $(G,k)$ with $(G\setminus(V(D)\setminus N_G(B)),k)$.
\end{RULE}
\begin{proof}[Proof of Safeness]
	Let $G':=G\setminus(V(D)\setminus N_G(B))$.
	If $(G,k)$ is a yes-instance, then there is a modulator $S'$ of size at most $k$.
	Then $G'\setminus(S'\setminus(V(D)\setminus N_G(B)))$ is also a $3$-leaf power, implying that $(G',k)$ is a yes-instance.
	
	We now show that if $(G',k)$ is a yes-instance, then so is $(G,k)$.
	Suppose that $G'$ has a modulator~$S'$ of size at most $k$ and $G\setminus S'$ has an obstruction~$H$.
	Since $G'\setminus S'$ is a $3$-leaf power, $H$ has at least one vertex in $V(D)\setminus N_G(B)$.
	Since $H$ is connected and $V(D)\cap N_G(S)=\emptyset$, at least one vertex of $H$ is in $V(D)\cap N_G(B)$.
	Thus,~$H$ has at least two vertices of $D$.
	Since $V(H)\cap S\neq\emptyset$, the set $V(H)\cap B$ is a clique cut-set of $H$, and therefore $H$ is not a hole.
	Thus, $\abs{V(H)}=5$.
	Since~$S$ is a $1$-redundant modulator of $G$, we have that $\abs{V(H)\cap S}=2$, and therefore $\abs{V(H)\cap B}=1$ and $\abs{V(H)\cap V(D)}=2$.
	Thus, the vertex $u\in V(H)\cap B$ is a cut-vertex of $H$ such that at least two components of~$H\setminus v$ have the same number of vertices, contradicting~(\ref{O4}).
\end{proof}

\subsection{(\ref{Rbagsecond}) is safe}\label{subsec:safeness6}

To show that (\ref{Rbagsecond}) is safe, we will use the following two lemmas.
Lemma~\ref{lem:completeness of bags} will be useful because it implies that for a $1$-redundant modulator $S$ of $G$, a set $B\subseteq V(G)\setminus S$ is a true twin-set in $G\setminus S$ if and only if it is a true twin-set in~$G$.

\begin{LEM}\label{lem:completeness of bags}
	Let $G$ be a $3$-leaf power.
	If $G$ has a vertex $v$ such that $G\setminus v$ is connected and incomplete, then 
	two vertices of $G\setminus v$ are true twins in $G\setminus v$ if and only if they are true twins in $G$.
\end{LEM}
\begin{proof}
	The backward direction is trivial.
	
	To prove the forward direction, 
	suppose for contradiction that $t_1$ and $t_2$ are true twins in~$G\setminus v$ while they are not true twins in $G$.
	Then $v$ is adjacent to exactly one of~$t_1$ and~$t_2$.
	By symmetry, we may assume that $v$ is adjacent to $t_1$.
	Note that $\abs{N_G(t_2)}\geq2$, because otherwise $G\setminus v$ is isomorphic to $K_2$.
	
	If $N_G(t_2)$ is a clique, then $G\setminus v$ has at least one vertex not in $N_G(t_2)$, because otherwise $G\setminus v$ is complete.
	Since $G\setminus v$ is connected, $G$ has an edge~$xy$ such that~$x$ is adjacent to both~$t_1$ and~$t_2$, and~$y$ is non-adjacent to both~$t_1$ and~$t_2$.
	Then $G[\{v,x,y,t_1,t_2\}]$ is isomorphic to the gem if both~$x$ and~$y$ are adjacent to~$v$, the bull if both~$x$ and~$y$ are non-adjacent to~$v$, and the dart if~$x$ is adjacent to~$v$ and~$y$ is non-adjacent to~$v$, and has a hole of length~$4$ if~$x$ is non-adjacent to~$v$ and~$y$ is adjacent to~$v$, contradicting the assumption that~$G$ is a $3$-leaf power.
	
	Therefore, $t_2$ has distinct neighbors~$x$ and~$y$ which are non-adjacent.
	Then $G[\{v,x,y,t_1,t_2\}]$ has a hole of length~$4$ if both~$x$ and~$y$ are adjacent to~$v$, and is isomorphic to the gem if exactly one of~$x$ and~$y$ is adjacent to~$v$, and the dart if both~$x$ and~$y$ are non-adjacent to~$v$, contradicting the assumption that~$G$ is a $3$-leaf power.
\end{proof}
	
\begin{LEM}\label{lem:obstruction outside}
	Let $G$ be a graph and $(A,B)$ be a complete split of $G$.
	If $G$ has an obstruction $H$ and a non-empty $1$-redundant modulator $S\subseteq B\setminus N(A)$, then $H$ has at most one vertex in $A$.
\end{LEM}
\begin{proof}
	Suppose not.
	Since $S$ is a $1$-redundant modulator of $G$, at least two vertices of $H$ are in $S$.
	Thus, $H$ has vertices in both $A$ and $B$.
	Since $(A,B)$ is a complete split of $G$, by Lemma~\ref{lem:no holes between}, $H$ is not a hole, and therefore $\abs{V(H)}=5$.
	Then $\abs{V(H)\cap N(A)}\leq5-\abs{V(H)\cap A}-\abs{V(H)\cap S}\leq5-2-2$.
	Since $H$ is connected, we have that $\abs{V(H)\cap N(A)}=1$ and $\abs{V(H)\cap A}=\abs{V(H)\cap(B\setminus N(A))}=2$.
	Thus, the vertex $u\in V(H)\cap N(A)$ is a cut-vertex of $H$ such that at least two components of $H\setminus v$ have the same number of vertices, contradicting (\ref{O4}).
\end{proof}

We show the safeness of (\ref{Rbagsecond}).
This rule will allow us to reduce the instance when a tree-clique decomposition of $G\setminus S$ for a $1$-redundant modulator~$S$ of a graph~$G$ has a node of large degree with certain properties.

\begin{RULE}[\ref{Rbagsecond}]
	Let $B$ be a maximal true twin-set in $G\setminus S$.
	If $G\setminus(S\cup B)$ has distinct components $D_1,\ldots,D_{k+4}$ such that $N_G(V(D_1))=\cdots=N_G(V(D_{k+4}))$, and either $V(D_1)\cup\cdots\cup V(D_{k+4})\subseteq N_G(B)$, or $\emptyset\neq V(D_i)\cap N_G(B)\neq V(D_i)$ for every $i\in\{1,\ldots,k+4\}$, then replace $(G,k)$ with $(G\setminus V(D_1),k)$.
\end{RULE}
\begin{proof}[Proof of Safeness]
	If $(G,k)$ is a yes-instance, then there is a modulator $S'$ of size at most $k$.
	Then $(G\setminus V(D_1))\setminus(S'\setminus V(D_1))$ is also a $3$-leaf power, implying that $(G\setminus V(D_1),k)$ is a yes-instance.
	
	We now show that if $(G\setminus V(D_1),k)$ is a yes-instance, then so is $(G,k)$.
	Suppose that $G\setminus V(D_1)$ has a modulator $S'$ of size at most~$k$ and $G\setminus S'$ has an obstruction~$H$.
	Since $G\setminus(V(D_1)\cup S')$ is a $3$-leaf power,~$H$ has at least one vertex of $D_1$.
	By the definition of the $1$-redundant modulator, $G\setminus(S\setminus\{v\})$ is a $3$-leaf power for every vertex $v\in S$.
	Thus, if $v$ has a neighbor in a true twin-set $X$ in $G\setminus S$, then $\{v\}$ is complete to $X$ by Lemma~\ref{lem:completeness of bags}.
	This means that every true twin-set in $G\setminus S$ is a true twin-set in $G$ as well.

	We will use the following claim.
	
	\begin{CLM}\label{clm:6-1}
		For every $i\in\{1,\ldots,k+4\}$, $V(D_i)\cap N_G(B)$ is a true twin-set in~$G\setminus S$.
	\end{CLM}
	\begin{subproof}
		We first show that $V(D_i)\cap N_G(B)$ is a clique.
		Suppose for contradiction that $V(D_i)\cap N_G(B)$ contains two vertices~$x$ and~$y$ which are non-adjacent.
		Since $D_i$ is connected, we can find an induced path~$P$ in~$D_i$ between~$x$ and~$y$.
		Since $D_i$ is a $3$-leaf power, by Lemma~\ref{lem:not DH}, the length of~$P$ is exactly~$2$.
		Let~$z$ be the internal vertex of~$P$.
		Then $z\in N_G(B)$, because otherwise $V(P)$ with a vertex in $B$ induces a hole of length $4$.
		Then for vertices $v\in B$ and $v'\in V(D_j)\cap N_G(B)$ for any $j\in\{1,\ldots,k+4\}\setminus\{i\}$, $G[\{v,v',x,y,z\}]$ is isomorphic to the dart, which is an induced subgraph of $G\setminus S$, contradicting the assumption that~$S$ is a modulator of $G$.
		Therefore, $V(D_i)\cap N_G(B)$ is a clique.
	
		We now need to show that $V(D_i)\cap N_G(B)$ is a twin-set in $G\setminus S$.
		Suppose for contradiction that $G\setminus S$ has a vertex $w$ adjacent to a vertex $t_1\in V(D_i)\cap N_G(B)$ and non-adjacent to a vertex $t_2\in V(D_i)\cap N_G(B)$.
		Note that $w$ is contained in $V(D_i)\setminus N_G(B)$.
		Then for vertices $v\in B$ and $v'\in V(D_j)\cap N_G(B)$ for any $j\in\{1,\ldots,k+4\}\setminus\{i\}$, $G[\{v,v',w,t_1,t_2\}]$ is isomorphic to the bull, which is an induced subgraph of $G\setminus S$ contradicting the assumption that~$S$ is a modulator of $G$.
	\end{subproof}
	
	By assumption, either $V(D_1)\cup\cdots\cup V(D_{k+4})\subseteq N_G(B)$, or $\emptyset\neq V(D_i)\cap N_G(B)\neq V(D_i)$ for every $i\in\{1,\ldots,k+4\}$ holds.
	Suppose first that $V(D_1)\cup\cdots\cup V(D_{k+4})\subseteq N_G(B)$.
	By Claim~\ref{clm:6-1} and~(\ref{O1}), for every $i\in\{1,\ldots,k+4\}$, at most one vertex of $H$ is in $V(D_i)$.
	If $H$ is a small obstruction, then by (\ref{O2}), at most three of $D_1,\ldots,D_{k+4}$ have vertices of $H$.
	Otherwise, at most two of $D_1,\ldots,D_{k+4}$ have vertices of $H$, because otherwise $H$ has a vertex of degree at least~$3$ in~$H$.
	Since $\abs{S'}\leq k$, one of $D_2,\ldots,D_{k+4}$, say~$D_j$, has no vertices in $S'\cup V(H)$.
	Let $s$ be the vertex in $V(H)\cap V(D_1)$ and $t$ be any vertex in $D_j$.
	Since $N_G(V(D_1))=N_G(V(D_j))$, we have that $s$ and $t$ have the same sets of neighbors in $V(H)$.
	Then $G[(V(H)\setminus\{s\})\cup\{t\}]$ is isomorphic to~$H$, which is an induced subgraph of $G\setminus(V(D_1)\cup S')$, contradicting the assumption that~$S'$ is a modulator of $G\setminus V(D_1)$.
	Hence, $\emptyset\neq V(D_i)\cap N_G(B)\neq V(D_i)$ for every $i\in\{1,\ldots,k+4\}$.
	
	We will use the following two claims.
	
	\begin{CLM}\label{clm:6-2}
		For every $i\in\{1,\ldots,k+4\}$, $V(D_i)\setminus N_G(B)$ contains no vertices in~$N_G(S)$.
	\end{CLM}
	\begin{subproof}
		Suppose for contradiction that $V(D_i)\setminus N_G(B)$ contains a vertex $p_i$ which is a neighbor of some $v\in S$.
		Let $j\in\{1,\ldots,k+4\}\setminus\{i\}$.
		Since $N_G(V(D_i))=N_G(V(D_j))$, $V(D_j)$ contains a vertex~$p_j$ which is a neighbor of~$v$.
		Since some vertex in $B$ has neighbors in both $D_i$ and $D_j$, by the connectedness, $G\setminus S$ has an induced path $P$ between $p_i$ and $p_j$.
		Note that the length of $P$ is at least $3$, because $p_i$ is not in $N_G(B)$.
		Since $v$ is adjacent to both ends of $P$, by Lemma~\ref{lem:not DH}, $G[V(P)\cup\{v\}]$ is not distance-hereditary, which is an induced subgraph of $G\setminus(S\setminus\{v\})$, contradicting the assumption that $S$ is $1$-redundant.
	\end{subproof}
	
	For every $i\in\{1,\ldots,k+4\}$, let $D_{i,1},\ldots,D_{i,m(i)}$ be the components of~$D_i\setminus N_G(B)$.
	
	\begin{CLM}\label{clm:6-3}
		Let $i\in\{1,\ldots,k+4\}$.
		For every $j\in\{1,\ldots,m(i)\}$, if $\abs{V(D_{i,j})}\geq2$, then $(V(D_{i,j}),V(G)\setminus V(D_{i,j}))$ is a complete split of $G$.
	\end{CLM}
	\begin{subproof}
		By Claims~\ref{clm:6-1} and~\ref{clm:6-2}, it suffices to show that $N_G(N_G(B))\cap V(D_{i,j})$ is a clique.
		Suppose for contradiction that $N_G(N_G(B))\cap V(D_{i,j})$ contains vertices $x$ and $y$ which are non-adjacent.
		Since~$D_{i,j}$ is connected, it has an induced path $P$ between $x$ and $y$.
		Since $D_{i,j}$ is a $3$-leaf power, by Lemma~\ref{lem:not DH}, the length of $P$ is exactly $2$.
		Let $z$ be the internal vertex of $P$.
		Then $z\in N_G(N_G(B))$, because otherwise $V(P)$ with a vertex $v\in N_G(B)\cap V(D_i)$ induces a hole of length~$4$.
		Then for any vertex $v'\in B$, $G[\{v,v',x,y,z\}]$ is isomorphic to the dart, which is an induced subgraph of $G\setminus S$, contradicting the assumption that $S$ is a modulator of $G$.
	\end{subproof}
	
	By Claim~\ref{clm:6-3} and Lemma~\ref{lem:obstruction outside}, for every $i\in\{1,\ldots,k+4\}$ and every $j\in\{1,\ldots,m(i)\}$, at most one vertex of $H$ is in $V(D_{i,j})$.
	By Claim~\ref{clm:6-1} and~(\ref{O1}), at most one vertex of $H$ is in $V(D_i)\cap N_G(B)$.
	Therefore, at most one component of $D_i\setminus N_G(B)$ has a vertex of $H$, because otherwise $H$ has false twins of degree at most $1$, contradicting~(\ref{O5}).
	By (\ref{O2}), if $H$ is a small obstruction, then at most three of $D_1,\ldots,D_{k+4}$ have vertices of $H$.
	Otherwise, at most two of $D_1,\ldots,D_{k+4}$ have vertices of $H$, because otherwise $H$ has a vertex of degree at least $3$ in~$H$.
	Since $\abs{S'}\leq k$, one of $D_2,\ldots,D_{k+4}$, say $D_i$, has no vertices in $S'\cup V(H)$.
	Since $H$ is connected and has vertices in both~$S$ and $V(D_1)$, by Claim~\ref{clm:6-2}, $H$ has a vertex $s_1\in V(D_1)\cap N_G(B)$.
	Let $t_1t_2$ be an edge of $D_i$ such that $t_1\in V(D_i)\cap N_G(B)$ and $t_2\in V(D_i)\setminus N_G(B)$.
	Since $N_G(V(D_1))=N_G(V(D_i))$, by Claim~\ref{clm:6-1}, $s_1$ and $t_1$ have the same sets of neighbors in $V(H)\setminus V(D_1)$.

	If $H$ has a vertex $s_2\in V(D_1)\setminus N_G(B)$, then $V(D_1)\cap V(H)=\{s_1,s_2\}$, because each of $V(D_1)\cap N_G(B)$ and $V(D_1)\setminus N_G(B)$ has at most one vertex of $H$.
	Then $G[(V(H)\setminus\{s_1,s_2\})\cup\{t_1,t_2\}]$ is isomorphic to $H$, which is an induced subgraph of $G\setminus(V(D_1)\cup S')$, contradicting the assumption that $S'$ is a modulator of $G\setminus V(D_1)$.
	
	Therefore, $H$ has no vertices in $V(D_1)\setminus N_G(B)$.
	Then $G[(V(H)\setminus\{s_1\})\cup\{t_1\}]$ is isomorphic to~$H$, which is an induced subgraph of $G\setminus(V(D_1)\cup S')$, contradicting the assumption that $S'$ is a modulator of $G\setminus V(D_1)$.
\end{proof}

\subsection{(\ref{Rbagthird}) is safe}\label{subsec:safeness7}

Finally we show the safeness of (\ref{Rbagthird}).
We will use the following lemma.
	
\begin{LEM}\label{lem:obstruction intersecting}
	Let $G$ be a graph.
	If $G$ has disjoint true twin-sets $B_1,\ldots,B_m$ for $m\geq5$ such that $N(B_i)=B_{i-1}\cup B_{i+1}$ for every $i\in\{2,\ldots,m-1\}$, then~$G$ is a $3$-leaf power if and only if $G\setminus(B_3\cup\cdots\cup B_{m-2})$ is a $3$-leaf power and has no paths between a vertex in $B_2$ and a vertex in~$B_{m-1}$.
\end{LEM}
\begin{proof}
	It is clear that if $G$ is a $3$-leaf power, then $G\setminus(B_3\cup\cdots\cup B_{m-2})$ is a $3$-leaf power and has no paths between a vertex in $B_2$ and a vertex in $B_{m-1}$, because otherwise $G$ has a hole.
	
	Conversely, suppose for contradiction that $G\setminus(B_3\cup\cdots\cup B_{m-2})$ is a $3$-leaf power and has no paths between a vertex in $B_2$ and a vertex in $B_{m-1}$ while~$G$ has an obstruction~$H$.
	Since $G\setminus(B_3\cup\cdots\cup B_{m-2})$ is a $3$-leaf power,~$H$ has at least one vertex in $B_3\cup\cdots\cup B_{m-2}$.
	For each $i\in\{1,\ldots,m\}$, since~$B_i$ is a true twin-set in $G$, by (\ref{O1}), at most one vertex of $H$ is in $B_i$.
	Then every vertex of $H$ in $B_2\cup\cdots\cup B_{m-1}$ has degree at most $2$ in~$H$.
	
	If $H$ has a vertex $v\in B_j$ for some $j\in\{3,\ldots,m-2\}$, then by (\ref{O6}), both~$B_{j-1}$ and~$B_{j+1}$ have vertices of $H$.
	This means that $B_i$ contains exactly one vertex of $H$ for every $i\in\{2,\ldots,m-1\}$.
	Then $H$ has vertices in each of~$B_1$ and~$B_m$ as well by (\ref{O6}).
	Thus, $V(H)\cap(B_2\cup\cdots\cup B_{m-1})$ contains at least three vertices of degree $2$ in~$H$, and therefore by (\ref{O7}), $H$ is a hole.
	However, $H\setminus(B_3\cup\cdots\cup B_{m-2})$ is a path in $G\setminus(B_3\cup\cdots\cup B_{m-2})$ between a vertex in~$B_2$ and a vertex in~$B_{m-1}$, a contradiction.
\end{proof}

We now show the safeness of (\ref{Rbagthird}), 
which allows us to shorten a long path in a tree-clique decomposition of $G\setminus S$ for a $1$-redundant modulator~$S$ of a graph~$G$.

\begin{RULE}[\ref{Rbagthird}]
	Let $B_1,\ldots,B_m$ be pairwise disjoint maximal true twin-sets in $G\setminus S$ for $m\geq6$ such that $N_G(B_i)=B_{i-1}\cup B_{i+1}$ for every $i\in\{2,\ldots,m-1\}$.
	Let $\ell$ be an integer in $\{3,\ldots,m-2\}$ such that $\abs{B_\ell}\leq\abs{B_i}$ for every $i\in\{3,\ldots,m-2\}$, and $G'$ be a graph obtained from $G\setminus((B_3\cup\cdots\cup B_{m-2})\setminus B_\ell)$ by making $B_\ell$ complete to $B_2\cup B_{m-1}$.
	Then replace $(G,k)$ with $(G',k)$.	
\end{RULE}
\begin{proof}[Proof of Safeness]
	We first show that if $(G,k)$ is a yes-instance, then so is~$(G',k)$.
	Suppose that $G$ has a minimal modulator $S'$ of size at most $k$.
	Since $S'$ is minimal, by Lemma~\ref{lem:minimal modulator}, $S'\cap B_i=\emptyset$ or $S'\cap B_i=B_i$ for every $i\in\{1,\ldots,m\}$.
	
	We first claim that if $S'\cap(B_1\cup\cdots\cup B_m)$ is empty, then $S'$ is a modulator of $G'$.
	Since $G\setminus S'$ is a $3$-leaf power, by Lemma~\ref{lem:obstruction intersecting}, $G\setminus(B_3\cup\cdots\cup B_{m-2}\cup S')$ is a $3$-leaf power and has no paths between a vertex in $B_2$ and a vertex in~$B_{m-1}$.
	Since $G\setminus(B_3\cup\cdots\cup B_{m-2}\cup S')$ is isomorphic to $G'\setminus(B_\ell\cup S')$, by Lemma~\ref{lem:obstruction intersecting}, $G'\setminus S'$ is a $3$-leaf power, and this proves the first claim.
	
	We second claim that if $S'\cap(B_1\cup B_2\cup B_{m-1}\cup B_m)$ is non-empty, then $S'\cap V(G')$ is a modulator of $G'$.
	Suppose for contradiction that $G'\setminus(S'\cap V(G'))$ has an obstruction~$H$.
	Since $G\setminus(B_3\cup\cdots\cup B_{m-2}\cup S')$ is a $3$-leaf power and is isomorphic to $G'\setminus(B_\ell\cup(S'\cap V(G')))$, we have that $G'\setminus(B_\ell\cup(S'\cap V(G')))$ is a $3$-leaf power.
	Therefore, at least one vertex of $H$ is in $B_\ell$.
	For every $i\in\{1,2,\ell,m-1,m\}$, since $B_i$ is a true twin-set in $G'$, by (\ref{O1}), at most one vertex of $H$ is in $B_i$.
	Then every vertex in $V(H)\cap(B_2\cup B_\ell\cup B_{m-1})$ has degree at most $2$ in~$H$.
	Thus, for every $i\in\{1,2,\ell,m-1,m\}$, by (\ref{O6}),~$B_i$ contains exactly one vertex of~$H$.
	Then $S'\cap V(G')$ contains at least one vertex of $H$, a contradiction, because $H$ is an induced subgraph of $G'\setminus(S'\cap V(G'))$, and this proves the second claim.
	
	Thus, we may assume that $S'\cap(B_1\cup B_2\cup B_{m-1}\cup B_m)$ is empty and $S'\cap(B_3\cup\cdots\cup B_{m-2})$ is non-empty.
	Let $T:=(S'\setminus(B_3\cup\cdots\cup B_{m-2}))\cup B_\ell$.
	Since $G\setminus(B_3\cup\cdots\cup B_{m-2}\cup S')$ is a $3$-leaf power and is isomorphic to $G'\setminus T$, we have that $G'\setminus T$ is a $3$-leaf power.
	Since $S'\cap(B_3\cup\cdots\cup B_{m-2})$ is non-empty,
	\begin{linenomath}
	\begin{align*}
		\abs{T}
		&=\abs{T\setminus(B_3\cup\cdots\cup B_{m-2})}+\abs{B_\ell}\\
		&\leq\abs{S'\setminus(B_3\cup\cdots\cup B_{m-2})}+\abs{B_\ell}\leq\abs{S'}\leq k.
	\end{align*}
	\end{linenomath}
	Therefore, if $(G,k)$ is a yes-instance, then so is $(G',k)$.	
	
	We now show that if $(G',k)$ is a yes-instance, then so is $(G,k)$.
	Suppose that $G'$ has a minimal modulator $S'$ of size at most $k$.
	Since $S'$ is minimal, by Lemma~\ref{lem:minimal modulator}, $S'\cap B_i=\emptyset$ or $S'\cap B_i=B_i$ for every $i\in\{1,2,\ell,m-1,m\}$.
	It suffices to show that $S'$ is a modulator of $G$.
	
	Since $G'\setminus S'$ is a $3$-leaf power, if $S'\cap(B_1\cup B_2\cup B_{m-1}\cup B_m)$ is empty, then by Lemma~\ref{lem:obstruction intersecting}, $G'\setminus(B_\ell\cup S')$ is a $3$-leaf power and has no paths from a vertex in $B_2$ to a vertex in $B_{m-1}$.
	Since $G'\setminus(B_\ell\cup S')$ is isomorphic to $G\setminus(B_3\cup\cdots\cup B_{m-2}\cup S')$, by Lemma~\ref{lem:obstruction intersecting}, $G\setminus S'$ is a $3$-leaf power.
	
	Thus, we may assume that $S'\cap(B_1\cup B_2\cup B_{m-1}\cup B_m)$ is non-empty.
	Suppose for contradiction that $G\setminus S'$ has an obstruction $H$.
	Since $G'\setminus(B_\ell\cup S')$ is a $3$-leaf power and is isomorphic to $G\setminus(B_3\cup\cdots\cup B_{m-2}\cup S')$, we have that $G\setminus(B_3\cup\cdots\cup B_{m-2}\cup S')$ is a $3$-leaf power.
	Therefore, at least one vertex of $H$ is in $B_3\cup\cdots\cup B_{m-2}$.
	For every $i\in\{1,\ldots,m\}$, since $B_i$ is a true twin-set in $G$, by (\ref{O1}), at most one vertex of $H$ is in $B_i$.
	Then every vertex in $V(H)\cap(B_2\cup\cdots\cup B_{m-1})$ has degree at most $2$ in~$H$.
	If $H$ has a vertex $v\in B_j$ for some $j\in\{3,\ldots,m-2\}$, then by (\ref{O6}), both~$B_{j-1}$ and~$B_{j+1}$ have vertices of~$H$.
	This means that $B_i$ contains exactly one vertex of $H$ for every $i\in\{2,\ldots,m-1\}$.
	Then $H$ has vertices in each of $B_1$ and $B_m$ as well by (\ref{O6}).
	Thus, $S'$ contains at least one vertex of $H$, a contradiction.
	Therefore, if $(G',k)$ is a yes-instance, then so is $(G,k)$.
\end{proof}

\section{A polynomial kernel}\label{sec:main proof}

In this section, we prove Theorem~\ref{thm:main}.
To do this, we analyze the size of an instance $(G,k)$ for the \LPD\ problem with $k>0$ to which none of the reduction rules are applicable.

In Subsection~\ref{subsec:number}, we will present two propositions to bound the number of components of $G\setminus S$ by a polynomial in $k+\abs{S}$, where $S$ is a $1$-redundant modulator of $G$.
In Subsection~\ref{subsec:size}, we will present two propositions to bound the size of each component of $G\setminus S$ by a polynomial in $k+\abs{S}$.
Finally, in Subsection~\ref{subsec:main}, we will prove our main theorem.

\subsection{The number of components}\label{subsec:number}

We first show that if (\ref{Rnumberfirst}) is not applicable to $(G,k)$, then the number of non-trivial components of $G\setminus S$ is $O(k\abs{S}^2)$.

\begin{PROP}\label{prop:number nontrivial}
	For an instance $(G,k)$ of the \LPD\ problem with $k>0$ and a non-empty $1$-redundant modulator $S$ of $G$, if (\ref{Rnumberfirst}) is not applicable to $(G,k)$, then $G\setminus S$ has at most $2(k+2)\abs{S}^2$ non-trivial components.
\end{PROP}

To prove Proposition~\ref{prop:number nontrivial}, we will use the following lemma.

\begin{LEM}\label{lem:partitioning}
	Let $G$ be a graph.
	Let $S$ be a $1$-redundant modulator of $G$ and $C$ be a component of $G\setminus S$.
	For a vertex $v\in S$, if $N_G(v)\cap V(C)$ contains distinct vertices $w_1$ and $w_2$ that are not true twins in $C$, then every component of $G\setminus S$ different from $C$ has no vertices in $N_G(v)$.
\end{LEM}
\begin{proof}
	Suppose that there is a component of $G\setminus S$ different from $C$ having a vertex $w$ which is a neighbor of $v$.
	Since $w_1$ and $w_2$ are not true twins in $C$, if $w_1$ and $w_2$ are adjacent, then $C$ has a vertex $w_3$ adjacent to exactly one of $w_1$ and $w_2$.
	Thus, $G[\{v,w,w_1,w_2,w_3\}]$ is isomorphic to the dart if~$w_3$ is adjacent to $v$, and the bull if $w_3$ is non-adjacent to $v$.
	However, the graph has exactly one vertex $v$ in $S$, contradicting the assumption that $S$ is a $1$-redundant modulator of $G$.
	
	Therefore, $w_1$ and $w_2$ are non-adjacent.
	Since $C$ is connected, there is an induced path~$P$ in~$C$ between $w_1$ and $w_2$.
	Since $S$ is a $1$-redundant modulator of $G$ and $v$ is adjacent to both ends of $P$, by Lemma~\ref{lem:not DH}, the length of $P$ is exactly $2$.
	Let $w_3$ be the internal vertex of $P$.
	Then $G[\{v,w,w_1,w_2,w_3\}]$ is isomorphic to the dart if $w_3$ is adjacent to $v$, and has a hole of length $4$ if $w_3$ is non-adjacent to $v$.
	However, the graph has exactly one vertex $v$ in $S$, contradicting the assumption that $S$ is a $1$-redundant modulator of $G$.
\end{proof}

We now prove Proposition~\ref{prop:number nontrivial}.

\begin{proof}[Proof of Proposition~\ref{prop:number nontrivial}]
	Let $S^+$ be the set of vertices $u\in S$ such that for each component $C$ of $G\setminus S$, $N_G(u)\cap V(C)$ is a true twin-set in $C$, and let $S^-:=S\setminus S^+$.
	By Lemma~\ref{lem:partitioning}, each vertex in~$S^-$ is adjacent to at most one component of $G\setminus S$.
	Therefore, $G\setminus S$ has at most $\abs{S^-}$ non-trivial components having neighbors of $S^-$.
	
	Let $Q$ and $M$ be defined as in (\ref{Rnumberfirst}).
	Since (\ref{Rnumberfirst}) is not applicable to $(G,k)$, each non-trivial component of $G\setminus S$ having no neighbors of $S^-$ is incident with exactly one edge in $M$.
	Since each edge in $M$ is incident with some element in $X\times\{1,2,3\}$ and each element in $X\times\{1,2,3\}$ is incident with at most $k+2$ edges,
	\[
		\abs{M}\leq(k+2)\cdot\abs{X\times\{1,2,3\}}\leq(k+2)\cdot3\binom{\abs{S^+}}{2}\leq3(k+2)\abs{S}^2/2.
	\]
	Thus,
	\begin{linenomath}
	\begin{align*}
		\abs{S^-}+\abs{M}
		&\leq\abs{S}+\frac{3(k+2)\abs{S}^2}{2}\\
		&\leq\frac{(k+2)\abs{S}^2}{2}+\frac{3(k+2)\abs{S}^2}{2}=2(k+2)\abs{S}^2,
	\end{align*}
	\end{linenomath}
	and therefore $G\setminus S$ has at most $2(k+2)\abs{S}^2$ non-trivial components.
\end{proof}

We show that if (\ref{Rnumbersecond}) is not applicable to $(G,k)$, then $G\setminus S$ has at most $O(k\abs{S}^4)$ isolated vertices.

\begin{PROP}\label{prop:number isolated}
	For an instance $(G,k)$ of the \LPD\ problem with $k>0$ and a non-empty $1$-redundant modulator $S$ of $G$, if (\ref{Rnumbersecond}) is not applicable to $(G,k)$, then $G\setminus S$ has at most $2(k+3)\abs{S}^4/3$ isolated vertices.
\end{PROP}
\begin{proof}
	Let $\mathcal{A}$, $X$, and $X_{A_1,A_2}$ be defined as in (\ref{Rnumbersecond}).
	Since $S$ is a $1$-redundant modulator of $G$, we may assume that $s:=\abs{S}\geq2$.
	For each subset~$T$ of~$S$ with $2\leq \abs{T}\leq4$, $\mathcal{A}$ contains exactly $2^{\abs{T}}$ elements $(A_1,A_2)$ such that $T=A_1\cup A_2$.
	Therefore, $\abs{\mathcal{A}}$ is at most
	\begin{linenomath}
	\begin{align*}
		2^4\cdot\binom{s}{4}+2^3\cdot\binom{s}{3}+2^2\cdot\binom{s}{2}&\leq\frac{2}{3}(s-1)^4+\frac{4}{3}(s-1)^3+2s(s-1)\\
		&=\frac{2}{3}(s-1)^2((s-1)^2+2(s-1))+2s(s-1)\\
		&\leq\frac{2}{3}(s-1)^2s^2+2s(s-1)\\
		&=2s(s-1)(s^2-s+3)/3\\
		&\leq2s(s-1)(s^2+s)/3=2s^2(s^2-1)/3\leq2s^4/3.
	\end{align*}
	\end{linenomath}
	
	For each element $(A_1,A_2)$ in $\mathcal{A}$, $\abs{X_{A_1,A_2}}\leq k+3$.
	Thus,
	\[
		\left|\bigcup_{(A_1,A_2)\in\mathcal{A}}X_{A_1,A_2}\right|\leq\frac{2(k+3)\abs{S}^4}{3}.
	\]
	Since (\ref{Rnumbersecond}) is not applicable to $(G,k)$, every isolated vertex of $G\setminus S$ is in $\bigcup_{(A_1,A_2)\in\mathcal{A}}X_{A_1,A_2}$.
	Therefore, $G\setminus S$ has at most $2(k+3)\abs{S}^4/3$ isolated vertices.
\end{proof}

\subsection{The size of each component}\label{subsec:size}

We show that if (\ref{Rcomplete}) is not applicable to $(G,k)$, then every complete component of $G\setminus S$ has at most $O(k\abs{S}^4)$ vertices.

\begin{PROP}\label{prop:size complete}
	For an instance $(G,k)$ of the \LPD\ problem with $k>0$ and a non-empty $1$-redundant modulator $S$ of $G$, if (\ref{Rcomplete}) is not applicable to $(G,k)$, then every complete component of $G\setminus S$ has at most $2(k+3)\abs{S}^4/3$ vertices.
\end{PROP}
\begin{proof}
	Let $\mathcal{A}$, $C$, and $X_{A_1,A_2}$ be defined as in (\ref{Rcomplete}).
	Since $S$ is a $1$-redundant modulator of $G$, we may assume that $\abs{S}\geq2$.
	For each subset~$T$ of~$S$ with $2\leq \abs{T}\leq4$, $\mathcal{A}$ contains exactly $2^{\abs{T}}$ elements $(A_1,A_2)$ such that $T=A_1\cup A_2$.
	Therefore, $\abs{\mathcal{A}}\leq2^4\cdot\binom{\abs{S}}{4}+2^3\cdot\binom{\abs{S}}{3}+2^2\cdot\binom{\abs{S}}{2}\leq2\abs{S}^4/3$, as in the proof of Proposition~\ref{prop:number isolated}.
	For each element $(A_1,A_2)$ in $\mathcal{A}$, $\abs{X_{A_1,A_2}}\leq k+3$.
	Thus, $\abs{\bigcup_{(A_1,A_2)\in\mathcal{A}}X_{A_1,A_2}}\leq2(k+3)\abs{S}^4/3$.
	Since (\ref{Rcomplete}) is not applicable to~$(G,k)$, every vertex of $C$ is in $\bigcup_{(A_1,A_2)\in\mathcal{A}}X_{A_1,A_2}$.
	Therefore, $C$ has at most $2(k+3)\abs{S}^4/3$ vertices.
\end{proof}

We show that if none of (\ref{Rnumberfirst}), (\ref{Rtruetwin}), (\ref{Rbagfirst}), (\ref{Rbagsecond}), and (\ref{Rbagthird}) is applicable to $(G,k)$, then every incomplete component of $G\setminus S$ has at most $O(k^2\abs{S}^2)$ vertices.

\begin{PROP}\label{prop:size incomplete}
	For an instance $(G,k)$ of the \LPD\ problem with $k>0$ and a non-empty $1$-redundant modulator $S$ of $G$, if none of (\ref{Rnumberfirst}), (\ref{Rtruetwin}), (\ref{Rbagfirst}), (\ref{Rbagsecond}), and (\ref{Rbagthird}) is applicable to $(G,k)$, then each incomplete component of $G\setminus S$ has at most $(k+1)(k+4)\abs{S}(\abs{S}+2k+15)$ vertices.
\end{PROP}

To prove Proposition~\ref{prop:size incomplete}, we will use the following lemma, which we slightly rephrase for our application.

\begin{LEM}[Brandst\"{a}dt and Le {\cite[Corollary~11]{BL2006}}]\label{lem:twin operation}
	Let $G$ be a connected $3$-leaf power 
	and $v$ be a vertex of~$G$.
	If $v$ is not a cut-vertex, then $G\setminus v$ has a true twin-set $B$ such that $N_G(v)=B$ or $N_G[v]=N_G[B]$.
\end{LEM}

\begin{proof}[Proof of Proposition~\ref{prop:size incomplete}]
	Let $C$ be an incomplete component of $G\setminus S$ with a tree-clique decomposition $(F,\{B_u:u\in V(F)\})$.
	Since $S$ is a $1$-redundant modulator of $G$, for every vertex $v\in S$, $G[V(C)\cup\{v\}]$ is a $3$-leaf power.
	Thus, if $S$ has a vertex $w$ having a neighbor in a bag $B$ of $C$, then by Lemma~\ref{lem:completeness of bags}, $\{w\}$ is complete to $B$.
	This means that every bag of $C$ is a true twin-set in $G$.
	Since (\ref{Rtruetwin}) is not applicable to $(G,k)$, each bag of $C$ contains at most $k+1$ vertices.
	Therefore, in the remaining of this proof, we are going to bound the number of bags of $C$.
	Let $X$ be the set of leaves of $F$ whose bags are anti-complete to $S$.
	
	\begin{CLM}\label{clm:last1}
		If a node $u$ of $F\setminus X$ has degree at most $1$ in $F\setminus X$, then $B_u\cap N(S)$ is non-empty.
	\end{CLM}
	\begin{subproof}[Proof of Claim~\ref{clm:last1}]
		If $N_F(u)\subseteq X$, then $B_u$ contains a neighbor of $S$, because otherwise $C$ has no neighbors of $S$ and (\ref{Rnumberfirst}) is applicable to $(G,k)$.
		If $N_F(u)\setminus X$ is non-empty, then $N_F(u)\setminus X$ contains exactly one node $u_1$, because $u$ has degree at most $1$ in $F\setminus X$.
		If $B_u$ contains no neighbors of~$S$, then (\ref{Rbagfirst}) is applicable to $(G,k)$ by taking $B_{u_1}$ as $B$.
		Therefore, $B_u$ contains a neighbor of $S$.
	\end{subproof}
	
	\begin{CLM}\label{clm:last2}
		The maximum degree of $F$ is at most $\abs{S}+2k+7$.
	\end{CLM}
	\begin{subproof}[Proof of Claim~\ref{clm:last2}]
		Suppose that $F$ has a node $u$ of degree at least $\abs{S}+2k+8$ in $F$.
		For each vertex $w$ in $S$, if at least two components of $C\setminus B_u$ have neighbors of $w$, then by Lemma~\ref{lem:twin operation}, all components of $C\setminus B_u$ have neighbors of $w$.
		Thus, for each vertex $w$ in $S$, we can choose a component of $C\setminus B_u$, say $D$, such that either all other components of $C\setminus B_u$ have neighbors of $w$, or no other components of $C\setminus B_u$ have neighbors of $w$.
		Since $C\setminus B_u$ has at least $\abs{S}+2k+8$ components, $C\setminus B_u$ has distinct components $D_1,\ldots,D_{2k+7}$ different from $D$ such that for every vertex $w\in S$, either all or none of them have neighbors of $w$.
		Thus, $N_G(V(D_1))=\cdots=N_G(V(D_{2k+7}))$.
		By the pigeonhole principle, $V(D_i)\subseteq N_G(B_u)$ or $\emptyset\neq V(D_i)\cap N_G(B_u)\neq V(D_i)$ is satisfied by at least $k+4$ values of $i$, contradicting the assumption that (\ref{Rbagsecond}) is not applicable to $(G,k)$.
	\end{subproof}
	
	For each vertex $v$ in $S$, let $X_v$ be the set of nodes of $F\setminus X$ whose bags contain neighbors of~$v$, let $S_1$ be the set of vertices $v$ in $S$ such that $X_v$ contains some leaf of $F\setminus X$, and let $S_2:=S\setminus S_1$.
	Note that by Lemma~\ref{lem:twin operation}, for every vertex $v\in S$, if $X_v$ is non-empty, then $F\setminus X$ has a node, say~$p$, such that $X_v=\{p\}$ or $X_v=N_{F\setminus X}[p]$.
	Let $F'$ be a tree obtained from $F\setminus X$ by contracting all edges in $F[X_v]$ for each vertex $v$ in $S$.
	By Claim~\ref{clm:last1},~$F'$ has at most $\abs{S_1}$ leaves, and therefore it has at most $\max(\abs{S_1}-2,0)$ branching nodes.
	Let $Y$ be the set of nodes of $F'$ which come from $X_v$ for some vertex $v\in S$, and let $Z$ be the set of branching nodes of $F'$.
	Then $\abs{Y\cup Z}\leq\abs{Y}+\abs{Z}\leq\abs{S}+\max(\abs{S_1}-2,0)\leq2\abs{S}$.
	Since (\ref{Rbagthird}) is not applicable to $(G,k)$, each component of $F'\setminus(Y\cup Z)$ has at most three nodes.
	Therefore, $\abs{V(F'\setminus(Y\cup Z))}\leq6\abs{S}$.
	Then by Claim~\ref{clm:last2}, $\abs{V(F\setminus X)}$ is at most
	\begin{linenomath}
	\begin{align*}
		\abs{Y}(\abs{S}+2k+8)+\abs{Z}+\abs{V(F'\setminus(Y\cup Z))}&\leq\abs{S}(\abs{S}+2k+8)+\abs{S}+6\abs{S}\\
		&=\abs{S}(\abs{S}+2k+15).
	\end{align*}
	\end{linenomath}
	Since (\ref{Rbagsecond}) is not applicable to $(G,k)$, each node of $F\setminus X$ is adjacent to at most $k+3$ nodes in~$X$.
	Thus, $\abs{V(F)}\leq(k+4)\abs{S}(\abs{S}+2k+15)$.
	By (\ref{Rtruetwin}), each bag of $C$ has at most $k+1$ nodes.
	Therefore, $\abs{V(C)}\leq(k+1)(k+4)\abs{S}(\abs{S}+2k+15)$.
\end{proof}

\subsection{A proof of the main theorem}\label{subsec:main}

We are now ready to prove Theorem~\ref{thm:main}.

\main*

\begin{proof}
	We first show that Algorithm~\ref{alg} is a polynomial-time algorithm.
	Since the algorithm of Proposition~\ref{prop:redundant} is a polynomial-time algorithm, it suffices to show whether each of (\ref{Rnumberfirst}), \ldots, (\ref{Rbagthird}) can be applied in polynomial time.
	Let $(G,k)$ be an instance of the \LPD\ problem with $k>0$ and $S$ be a $1$-redundant modulator of $G$ obtained by Proposition~\ref{prop:redundant}.
	
	For every pair of vertices in $S$ and a set $X\subseteq V(G)$, one can in polynomial time check whether the pair is a blocking pair for $X$.
	Thus, the bipartite graph $Q$ in (\ref{Rnumberfirst}) can be constructed in polynomial time.
	By greedily choosing edges, we can find a maximal $(X\times\{1,2,3\},k+2)$-matching, so that one can apply (\ref{Rnumberfirst}) in polynomial time.
	As shown in the proof of Proposition~\ref{prop:number isolated}, the size of $\mathcal{A}$ is at most $2\abs{S}^4/3$, which is a polynomial in $\abs{V(G)}$.
	For each element $(A_1,A_2)\in\mathcal{A}$, by greedily packing vertices, we can compute $X_{A_1,A_2}$ in polynomial time.
	Therefore, (\ref{Rnumbersecond}) can be applied in polynomial time.
	Similarly, one can apply (\ref{Rcomplete}) in polynomial time.
	One can in polynomial time enumerate all maximal true twin-sets in~$G$ or $G\setminus S$.
	Thus, one can apply (\ref{Rtruetwin}), (\ref{Rbagfirst}), (\ref{Rbagsecond}), and (\ref{Rbagthird}) in polynomial time.
	Therefore, Algorithm~\ref{alg} is a polynomial-time algorithm.
		
	We now claim that if $(G',k')$ is the instance returned in Line~\ref{Return}, then~$G'$ has at most $O(k'^2\abs{S}^6)$ vertices.
	By Proposition~\ref{prop:size complete}, each complete component of $G'\setminus S$ has at most $2(k'+3)\abs{S}^4/3$ vertices.
	By Proposition~\ref{prop:size incomplete}, each incomplete component of $G'\setminus S$ has at most $(k'+1)(k'+4)\abs{S}(\abs{S}+2k'+15)$ vertices.
	Therefore, each non-trivial component of $G'\setminus S$ has at most $O(k'\abs{S}^4)$ vertices.
	By Proposition~\ref{prop:number nontrivial}, $G'\setminus S$ has at most $2(k'+2)\abs{S}^2$ non-trivial components.
	Thus, the union of all non-trivial components of $G'\setminus S$ has at most $2(k'+2)\abs{S}^2\cdot O(k'\abs{S}^4)=O(k'^2\abs{S}^6)$ vertices.
	By Proposition~\ref{prop:number isolated}, $G'\setminus S$ has at most $2(k'+3)\abs{S}^4/3$ isolated vertices.
	Thus, $\abs{V(G')}\leq\abs{S}+2(k'+3)\abs{S}^4/3+O(k'^2\abs{S}^6)=O(k'^2\abs{S}^6)$.
	By Line~\ref{k}, $k=k'$ and therefore 
	by Proposition~\ref{prop:redundant}, $\abs{S}\le 84k'^2+7k'$.
	We conclude that$\abs{V(G')}=O(k'^{14})$.
\end{proof}

\section{Conclusions}\label{sec:conclu}

We show that the \LPD\ problem admits a kernel with $O(k^{14})$ vertices.
It would be interesting to significantly reduce the size of the kernel or present a lower bound for the size of the kernel of the \textsc{$\ell$-Leaf Power Deletion} for any $\ell\geq3$, under some complexity hypothesis.

Gurski and Wanke~\cite{GW2009} stated that $\ell$-leaf powers have bounded clique-width for every non-negative integer $\ell$.
Rautenbach~\cite{Rautenbach2006} presented a characterization of $4$-leaf powers having no true twins as chordal graphs with ten forbidden induced subgraphs.
This can be used to express, in monadic second-order logic, whether a graph is a $4$-leaf power and whether there is a vertex set of size at most $k$ whose deletion makes the graph a $4$-leaf power.
Therefore, by using the algorithm in \cite{CMR2000}, the \textsc{$4$-Leaf Power Deletion} problem is fixed-parameter tractable when parameterized by $k$.
It is natural to ask whether the \textsc{$4$-Leaf Power Deletion} problem admits a polynomial kernel.
For $\ell\geq5$, we do not know whether we can express $\ell$-leaf powers in monadic second-order logic.
If it is true for some $\ell$, then the \textsc{$\ell$-Leaf Power Deletion} problem is fixed-parameter tractable.

\end{document}